\documentclass{elsarticle}
\journal{Computational Geometry: Theory and Applications}

\usepackage{fullpage}

\usepackage[ruled]{algorithm2e}
\usepackage{graphicx}
\usepackage{amsmath}
\usepackage{amssymb}
\usepackage[american]{babel}
\usepackage{dsfont}
\usepackage{todonotes}
\usepackage{xspace}
\usepackage{url}
\usepackage[pdfborder={0 0 0}]{hyperref} % last \usepackage
\usepackage{subfig}

\newtheorem{theorem}{Theorem}
\newtheorem{lemma}[theorem]{Lemma}

\newtheorem{definition}{Definition}
\newtheorem{claim}{Claim}

\def\QED{\ensuremath{{\Box}}}
\def\markatright#1{\leavevmode\unskip\nobreak\quad\hspace*{\fill}{#1}}
\newenvironment{proof}
 {\begin{trivlist}\item[\hskip\labelsep{\bf Proof.}]}
 {\markatright{\QED}\end{trivlist}}

\vfuzz \hfuzz

\SetAlFnt{\small}
\SetAlCapFnt{\small}
\SetAlCapNameFnt{\small}
\SetAlCapHSkip{0pt}
\IncMargin{-\parindent}

\newcommand{\old}[1]{{}}
\newcommand{\dist}{{\mbox{dist}}}

\begin{document}

\begin{frontmatter}

% Page heads
%\markboth{E.W.~Chambers et al.}{Connecting a Set of Circles with Minimum Sum of Radii}

\title{Connecting a Set of Circles with Minimum Sum of Radii}

%% \tnotetext[label1]{}
\author[stlu]{Erin W.~Chambers}
\ead{echambe5@slu.edu}
\author[tubs]{S{\'a}ndor P.\ Fekete\corref{corr}}
\ead{s.fekete@tu-bs.de}
\author[uw]{Hella-Franziska Hoffmann}
\ead{hrhoffma@uwaterloo.ca}
\author[kins]{Dimitri Marinakis}
\ead{dmarinak@kinsolresearch.com}
\author[usb]{Joseph~S.~B.~Mitchell}
\ead{jsbm@ams.sunysb.edu}
\author[uvic]{Venkatesh Srinivasan}
\ead{srinivas@uvic.ca}
\author[uvic]{Ulrike Stege}
\ead{ustege@uvic.ca}
\author[uvic]{Sue Whitesides}
\ead{sue@uvic.ca}

\address[stlu]{Department of Computer Science, Saint Louis University, USA}
\address[tubs]{Department of Computer Science, TU Braunschweig, Germany}
\address[uw]{David R. Cheriton School of Computer Science, University of Waterloo, Canada}
\address[kins]{Kinsol Research Inc., Duncan, BC, Canada}
\address[usb]{Department of Applied Mathematics and Statistics, Stony Brook University, USA}
\address[uvic]{Department of Computer Science, University of Victoria, Canada}

\cortext[corr]{Corresponding author}
%\fntext[fn1]{A preliminary extended abstract summarizing some of this work appeared in the Algorithms and Data Structures Symposium (WADS), 2011.
%This work has been partially supported by NSF grants CCF-1054779 and IIS-1319573 (Erin Chambers), by grants from the National Science Foundation (CCF-1018388), 
%the Binational Science Foundation (BSF 2010074), and Sandia National Labs (Joseph Mitchell), and three individual NSERC Discovery grants (one each
%for Venkatesh Srinivasan, Ulrike Stege, and Sue Whitesides).}

% NOTE! Affiliations placed here should be for the institution where the
%       BULK of the research was done. If the author has gone to a new
%       institution, before publication, the (above) affiliation should NOT be changed.
%       The authors 'current' address may be given in the "Author's addresses:" block (below).
%       So for example, Mr. Abdelzaher, the bulk of the research was done at UIUC, and he is
%       currently affiliated with NASA.

\begin{abstract}
We consider the problem of assigning radii to a given set of %facilities, such that
%the resulting set of ranges is connected, and the sum of radii is minimized.
%A natural special case is the situation in which the facilities are represented by 
points in the plane, such that the resulting set of disks is connected, and the sum of radii
is minimized. We prove that the problem is NP-hard in planar weighted graphs if there are upper bounds on the
radii and sketch a similar proof for planar point sets.
% xxx   REMOVED (for now):  and provide a polynomial-time approximation scheme. 
For the case when there are no upper bounds
on the radii, the complexity is open; we give a polynomial-time approximation scheme.
%\old{, but NP-hard if there are upper bounds on the radii; the case of unbounded radii is an open problem.}
We also give constant-factor approximation guarantees for solutions with a bounded number of disks; these results
are supported by lower bounds, which are shown to be tight in some of the cases. Finally, we show that the problem is polynomially solvable if
a connectivity tree is given, and we conclude with some experimental results.

\end{abstract}


\end{frontmatter}

%\category{F.2.2}{Nonnumerical Algorithms and Problems}{Geometrical problems and computations}

%\terms{Algorithms}

\section{Introduction}
Problems of connectivity are among the most fundamental ones for many types of networks.
Typically, these arise in a geometric setting, e.g., when considering
the relation between the location
of transmitters, the range of their transmissions, and their ability to connect.
As a result, important aspects include the underlying geometry, the nature of connectivity,
and the study of corresponding cost functions.
%---not necessarily just for accurately modeling
%the precise physical behavior of any real-world wireless sensor network, but also for
%understanding the algorithmic consequences of using an abstract theoretical model for 
%network connectivity. 
Thus, connectivity problems bring together graph theory with computational geometry,
a combination that was always dear to Ferran Hurtado; e.g., see \cite{agh+-acgg-08} for
a study of connectivity that mixes graphs and geometry.

In this paper, we consider a natural connectivity problem, 
with a focus on the geometric aspects, arising from 
assigning ranges to a set of points, such that the resulting disk intersection graph is
connected. More precisely, we are given 
%a finite set $P$ of objects and the pairwise
%distances of the elements in $P$ measured by metric $d$. In the graph setting,
%we are given a weighted graph $G=(V,E)$ where $V = P$ and the metric $d$ is induced by shortest path lengths. In 
%the geometric setting, $P$ is assumed to be 
%a finite set of points in the Euclidean plane.
%In both of these settings, 
%each point $p$ in $P$ is assigned a range $r_p$,
%and two vertices $v$ and $w$ are connected by an edge $f_{vw}$ in the connectivity graph
%$H=(V,F)$, if the shortest-path distance $d(v,w)$ in $G$ does not exceed the
%sum $r_v+r_w$ of their assigned radii. In a geometric setting,
%$V$ is given as 
a set of points $P = \{p_1,\ldots,p_n\}$ in the plane. Each point $p_i$ is assigned a range $r_i$,
inducing a disk of radius $r_i$.
% correspond to circular ranges:
Two points $p_i$, $p_j$ are adjacent in the connectivity graph $H$,
if their disks intersect.
The {\sc Connected Range Assignment Problem} (CRA) requires an assignment of radii
to $P$, such that the objective function $R = \sum_i r_i^\alpha, \alpha=1$ is
minimized, subject to the constraint that $H$ is connected.

Problems of this type have been considered before and have natural motivations from
fields including networks, robotics, and data analysis, where 
ranges have to be assigned to a set of devices, and the total cost is given by 
an objective function that considers the sum of the
radii of disks to some exponent $\alpha$. The cases $\alpha=2$ or $3$ correspond to
minimizing the overall power. The motivation for the case $\alpha=1$ arises from
scanning the corresponding ranges with a minimum required angular resolution,
so that the scan time for each disk corresponds to its perimeter, and thus radius.

\subsection{Related Work}
There is a large body of literature on algorithmic methods for range assignment
in wireless sensor and ad-hoc networks; see 
Calinescu {\em et al.}~\cite{cmwz-sfnwn-01},
Calinescu and Wan~\cite{SymAsym},
Carmi and Katz~\cite{ck-patpl-04},
Carmi {\em et al.}~\cite{CKSS},
Caragiannis {\em et al.}~\cite{CKK06},
Lloyd {\em et al.}~\cite{LLMRR},
Wan {\em et al.}~\cite{WCLF}
for a (small) selection of aspects and variants.

In the context of clustering, Doddi \emph{et al.}~\cite{Doddi00_apxMinSumDiam}, 
Charikar and Panigraphy~\cite{Charikar04_minSumDiam}, and 
Gibson \emph{et al.}~\cite{Gibson08_minSumRadii} consider the following problems.
Given a set $P$ of $n$ points in a metric space with metric 
$d(i,j)$ and an integer $k$, partition $P$ into a set of at most $k$ clusters with
minimum sum of either (a) cluster diameters, or (b) cluster radii.  Thus, the most significant
difference with our problem is the lack of a connectivity constraint.
Doddi \emph{et al.}~\cite{Doddi00_apxMinSumDiam} provide approximation results for (a).
They present a polynomial-time algorithm, which returns $O(k)$ clusters that
are $O(\log(\frac{n}{k}))$-approximate. For a fixed $k$, they transform an instance of
their problem into a min-cost set-cover problem instance yielding a polynomial-time $2$-approximation.
They also show that the existence of a $(2-\epsilon)$-approximation would imply
$P=NP$. In addition, they prove that the problem in weighted graphs without
triangle inequality cannot be efficiently approximated within any factor, unless
$P=NP$. Note that every solution to (b) is a $2$-approximation for (a). Thus, the
approximation results can be applied to case (a) as well.
A greedy logarithmic approximation and a primal-dual based constant factor
approximation for minimum sum of cluster radii is provided by 
Charikar and Panigraphy~\cite{Charikar04_minSumDiam}. 

Another often considered setting is the one in which the coverage of a given set of base stations is required. 
In this setting, Alt \emph{et al.}~\cite{aab+-mccps-06} consider
a closely related problem of selecting disk centers and radii such that a
given set of points in the plane are covered by the disks. Like our work, they
focus on minimizing an objective function based on $\sum_i r_i^{\alpha}$ and
produce results specific to various values of $\alpha$.  The minimum sum of
radii disk coverage problem (with $\alpha = 1$) is also considered by 
Lev-Tov and Peleg~\cite{Lev-Tov05} in the context of radio networks.  Again, connectivity
is not a requirement. In a more geometric setting,
Bil{\`o} \emph{et al.}~\cite{Bilo05_geomClustering} provide approximation schemes
for the minimum size $k$-clustering problem that requires dividing the set of centers
into at most $k$ clusters with minimum cluster cost.

A lot of work has also been done on radii/range assignment problems which require the special connectivity of  ``communication''.
The work of Clementi \emph{et al.}~\cite{Clementi04_powerAssignments} considers 
minimal assignments of transmission power to devices in a wireless network such that the network stays connected.
In that context, the objective function typically considers an $\alpha > 1$ based on models of radio wave propagation.
Furthermore, in the type of problem considered by Clementi \emph{et al.} the connectivity graph is directed;
\emph{i.e.} the power assigned to a specific device affects its transmission range, but not its reception range.
This is in contrast to our work in which we consider an undirected connectivity graph.
See Fuchs~\cite{Fuchs08_hardnessRAP} for a collection of hardness results
of different (directed) communication graphs.
Carmi \emph{et al.}~\cite{Carmi05_MAST} prove that an Euclidean minimum spanning tree
is a constant-factor approximation for a variety of problems including the {\em
Minimum-Area Connected Disk Graph} problem, which equals our problem with the
different objective of minimizing the {\em area} of the {\em union} of disks,
while we consider minimizing the {\em sum} of the {\em radii} (or perimeters) of all disks.
%A similar connectivity is assumed in some
%variants of the power assignment problem in radio networks.
This can apply to a hybrid robot or sensor
network system such as that described in 
Marinakis {\em et al.}~\cite{Marinakis07_aaaiHybridSelfLoc}:
Consider a situation in which a mobile robot is required to visit the region of each static network component,
\emph{e.g.,} for environmental monitoring purposes, and requires constant one-way information from at least
one static network component at all times; \emph{e.g.,} for navigational purposes.
In this case, a reasonable objective is to reduce the area covered by
the network while maintaining connectivity.

\subsection{Our Work}
In this paper we investigate a variety of algorithmic aspects of the CRA problem.
In Section~\ref{sec:cract}, we show that for a given connectivity tree, an optimal solution
can be computed efficiently. Section~\ref{sec:np-proof} gives a proof of NP-hardness for the
problem when there is an upper bound on the radii, with full details for the case of planar weighted graphs
and a sketch for planar point sets. Section~\ref{sec:lim} provides a number of
approximation results for solutions with bounded number of disks. 
In Section~\ref{sec:ptas}, we present a PTAS for the case in which there are no upper bounds on the radii of the disks.
%bounded case with upper bounds on the radii and for the unbounded case. 
These theoretical results are complemented by experimental
results in Section~\ref{sec:exp}. A concluding discussion with open problems is provided in Section~\ref{sec:conc}.

\section{CRA for a Given Connectivity Tree}\label{sec:cract}

In some settings, a specific set of connectivity edges is required instead of {\em any} connectivity tree.
In this section, we will show that the Connected Range Assignment problem is solvable in
polynomial time, when we are given a specific connectivity tree.
First we give a formal definition of this variant.

In the {\em Connected Range Assignment Problem for a Given Connectivity Tree (CRACT)}, we are given a finite set $P$, a metric $d$ and a spanning tree $T$ for the vertex set $P$.
We are required to find an assignment $r_i$ of radii for point $p_i$ in $P$ such that $R = \sum_i  r_i$ is minimized subject to the constraint that if $(p_i,p_j)$ is an edge in $T$, then the two corresponding disks $C_i$ and $C_j$ (with radii $r_i$ and $r_j$ respectively) intersect.

%\noindent
%{\sc Connected Range Assignment for a Given\\ Connectivity Tree (CRACT)}\label{prob:p1}
%	\begin{description}
%		\item[Input: ] A finite set $P$, a metric $d$ and a spanning tree $T$ for vertex set $P$.
%		\item[Output:] An assignment $r$ of radii for each point in $P$ such that
%		\begin{enumerate}
%			\item If $(p_i,p_j)$ is an edge in $T$, then the two corresponding circles $C_i$ and $C_j$ (with radii $r_i$ and $r_j$ respectively) intersect and
%			\item $R = \sum\limits_{p_i \in P} r_i$ is minimized.
%		\end{enumerate}
%	\end{description}
	
We call a solution for Problem CRACT an {\em optimal range assignment for $T$}.
We first state some structural properties of optimal range assignments and then, based on those,
provide a polynomial-time algorithm for Problem CRACT.
In the following, we always consider the case where $|P|$ is at least $2$ and assume that $T$ is rooted at some internal node $p_r$ of $T$.

\begin{lemma}\label{lem:zero_leaves}
Given a connectivity tree $T$ with at least three nodes,
there exists an optimal range assignment for $T$ with $r_l = 0$ for all leaves $p_l$ of $T$.
\end{lemma}
\begin{proof}
Consider a CRACT instance $(P,d,T)$.
Assume an optimal range assignment $r$ for $T$ has a leaf $p_l \in P$ with $r_l > 0$.
Then, the disk $C_l$ around $p_l$ with radius $r_l$ must intersect the disk
$C_u$ around $p_l$'s parent $p_u$ with radius $r_u$.
We extend $r_u$ to $r_u := r_u + r_l$ and set $r_l := 0$.
This does not increase $R = \sum_{p_i\in P} r_i$ and maintains a solution for $T$.
\end{proof}
%
%Considering the radii assigned to points of height $1$ (denoted $p_p$ in the former proof),
%we get the following direct consequence of Lemma~\ref{lem:zero_leaves}.
%%
%\begin{corollary} \label{cor:CRACT}
%There is an optimal range assignment satisfying Lemma~\ref{lem:zero_leaves} in 
%which $r_p > 0$ for every $p_p \in P$ of height $1$ in $T$ (i.e., each $p_p$
%is a parent of leaves only).
%\end{corollary}

The following lemma extends Lemma~\ref{lem:zero_leaves}. It shows that there is an optimal range assignment for which it is possible to compute the radius assigned to any point $p$ of  height $h$, given the radius assigned to its children.

\begin{lemma} \label{lem:parent_max}
Given a CRACT instance $(P,d,T)$, there is an optimal range assignment $r$ for $T$ that satisfies the following conditions:
\begin{enumerate}
\item
$r_l = 0$ for all leaves $p_l$ of $T$
\item
For any node $p_u \in P$ of height h, $1 \leq h \leq h_T$,  in $T$,
$$r_u = \max\limits_{p_c {\mbox{ \scriptsize is child of } p_u}} d(p_c,p_u)-r_c$$
\end{enumerate}
\end{lemma}

\begin{proof}

To prove that $r$ satisfying the two conditions above is an optimal range assignment, we show that any optimal range assignment can be gradually converted into assignment $r$ in a sequence of $h_T+1$ steps as follows: At the end of step $i$,  all nodes $p_u$ in $T$ with height $h$, $0 \leq h \leq i-1$, will be assigned the radius $r_u$. Furthermore, the sum of radii at the end of step $i$ is at most the sum of radii at the end of step $i-1$. 

[Step 1] Perform the transformation as described in the proof of Lemma~\ref{lem:zero_leaves}.  It is easily seen that  the condition (1) is true at the end of step 1. 

[Step $i$] Let $p_u$ be a node of height $i-1$ in $T$. Further let 
 $d_{max}(p_u) := \max d(p_c,p_u)$ over all children $p_c$ of $p_u$.
 Then we have $r_u \geq d_{max}(p_u)$.
 If $r_u > d_{max}(p_u)$, we shrink $C_u$ until it just covers its farthest child
 and extend the circle of $p_u$'s unique parent (height~$\geq i$) by the same amount. This ensures that condition (2) is satisfied for the point $p_u$. Furthermore,
 this maintains all connections induced by $C_u$ before shrinking and does not increase
 the solution value $R$.
\end{proof}

Using the observed optimal solution properties, we next describe an algorithm
which will result in an optimal range assignment for $T$.

\begin{figure}[h]
	\centering
\begin{minipage}{0.45\textwidth}
 	\includegraphics[width=\textwidth]{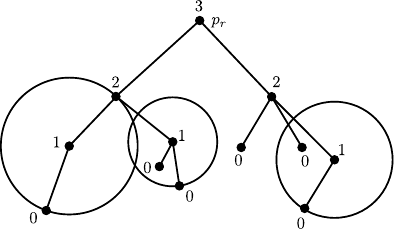}
\end{minipage}
\hfill
\begin{minipage}{0.45\textwidth}
 	\includegraphics[width=\textwidth]{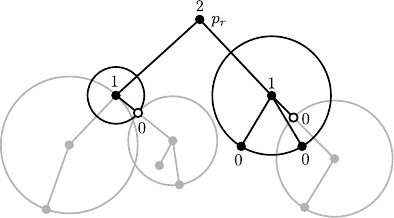}
\end{minipage}
	\caption[Sketch of Algorithm~\ref{alg:cract}.]{One step in algorithm~\ref{alg:cract}
	  for a tree $T$ of height $h_T$ (left).
	  The tree $\tilde{T}$ of height $h_T-1$ resulting from the remaining gaps is shown on the right. 
	  \label{fig:cract}}
\end{figure}

\begin{algorithm}
	\SetKwFor{ForAll}{forall}{do}{end}
	\SetKwFunction{Union}{Union}\SetKwFunction{FindCompress}{FindCompress}
	\SetKwInOut{Input}{Input}\SetKwInOut{Output}{Output}
	\SetAlgoLined

	\Input{A CRACT instance $(P,d,T)$, where $T = (P,E_T)$ of height $h_T$ with weighted edges: $d: E_T \rightarrow \mathbb{R}_+$.}
	\Output{An optimal range assignment $r: P \rightarrow \mathbb{R}_+$ for $T$}
	\BlankLine
	$h:=0$\;
	\lForAll{leaves $p_l$ in $T$}{$r_l:=0$\;}
	\While{$h\leq h_T$}{
		\ForAll{nodes $p_u$ of height $h$ in $T$}{
			$r_u:= \max\limits_{p_c\text{ child of } p_u}{\{d(p_c,p_u)-r_c\}}$\;
		}
		$h := h+1$\;
	}
        \Return{$r$}
	\caption{Algorithm solving $CRACT$\label{alg:cract}}
\end{algorithm}

\newpage
\begin{theorem} \label{t1}
CRACT  can be solved in $O(n)$ time.
\end{theorem}

\begin{proof}
	{\em 1) [Correctness].} In the range assignment $r$ computed by Algorithm~\ref{alg:cract}, a node $p_u$ at height $h$, $0 < h \leq h_T$,  in $T$ is assigned 
	a radius $r_u$ where  $r_u:= \max\limits_{p_c\text{ child of } p_u}{\{d(p_c,p_u)-r_c\}}$ and the leaves $p_l$ are assigned the radius $r_l=0$. The correctness of the algorithm now follows from Lemma~\ref{lem:parent_max}.
%	In the first and second iteration, the algorithm clearly 
%	  provides smallest possible radii for points of height $0$ and $1$ respectively as is evident from
%	  Lemma~\ref{lem:zero_leaves}~\&~\ref{lem:parent_max}. Furthermore,
%	  a solution that maintains these properties and establishes the remaining connections
%	  with minimum sum of radii will be an optimal solution to CRACT for (the whole tree) $T$.
%	  Inductively applying Lemma~\ref{lem:zero_leaves}~\&~\ref{lem:parent_max} to the
%	  remaining tree $\tilde{T}$ (see Fig.~\ref{fig:cract}), we conclude that the algorithm
%	  in each iteration computes a part of optimal solution to $\tilde{T}$, without changing
%	  the previously set radii.
	  
	{\em 2) [Runtime].} As the height of any node in $T$ is unique for a fixed root,
	  the algorithm assigns radii to each node in the tree exactly once. For every node $p$
	  of height $h \geq 1$, the computation of the required radius needs time proportional
	  to its degree. During this computation, each edge in the
	  tree (of which we have exactly $n-1$) is considered exactly once in total for
	  the computation. Hence, the algorithm determines an optimal
	  range assignment in $O(n)$ steps.
	  Note that even if the height of the nodes cannot be gained directly from the
	  data structure, it can be computed via one iteration over the tree,
	   which again only requires linear time.
\end{proof}

%For a given connectivity tree, our
%problem is polynomially solvable, based on the following observation.
%
%\begin{lemma}\label{l1}
%Given a connectivity tree $T$ with at least three nodes.
%There exists an optimal range assignment for 
%$T$ with $r_i = 0$ for all leaves $p_i$ of $T$.
%\end{lemma}
%%
%\begin{proof}
%Assume an optimal range assignment for $T$ has a leaf $p_i \in P$ with radius
%$r_i > 0$.  The circle $C_i$ around $p_i$ with radius $r_i$ intersects circle
%$C_j$ around $p_i$'s parent $p_j$ with radius $r_j$.
%Extending $C_j$ to $r_j := \dist(p_i,p_j)$ while setting $r_i := 0$ does not increase the solution value $R = \sum_{p_i\in P} r_i$. \qed
%\end{proof}
%%
%Direct consequences of Lemma~\ref{l1} are the following.
%%
%\begin{corollary} \label{c1}
%There is an optimal range assignment satisfying Lemma~\ref{l1} and $r_j > 0$ for all $p_j \in P$ of height $1$ in $T$ (i.e., each $p_j$ is a parent of leaves only).
%\end{corollary}
%%
%\begin{corollary} \label{c2}
%Consider an optimal range assignment for $T$ satisfying Lemma~\ref{l1}.  Further let $p_j \in P$ be of height 1 in $T$.
%Then $r_j \geq \max_{p_i {\mbox{ \scriptsize is child of } p_j}} \{\dist(p_i,p_j)\}$.
%\end{corollary}
%%
%These observations allow a solution by dynamic programming.
%The idea is to compute the values for subtrees, starting from the leaves.
%Details are omitted.
%%
%\begin{theorem} \label{t1}
%For a given connectivity tree, CRA is solvable in $O(n)$.
%\end{theorem}
%%

\section{NP-hardness for Bounded Radii}\label{sec:np-proof}

In this section, we show that, if an upper bound of $\beta$ on the radii is also specified,  the
problem CRA becomes NP-hard using a reduction from 3-SAT. More formally, we are given a finite set $P$ of objects, the pairwise
distances of the elements in $P$ measured by metric $d$ and an upper bound $\beta$ on the radii.  The goal is to find an assignment of radii
to $P$, such that the objective function $R = \sum_i r_i$ is
minimized, subject to the constraint that $H$ is connected and $r_i \leq \beta$ for all $i$.

\begin{definition}[3-SAT]
	In logic, a {\em literal} is a single (possibly negated) variable.
	A disjunction of three literals is called a {\em clause} $\gamma_i$.
	Given a Boolean formula ${\cal F} = \gamma_1 \wedge \gamma_2 \wedge \ldots \wedge \gamma_c$ 
	that is a conjunction of $c$ clauses formed by a set of $v$ variables,
	the $3$-Satisfiability problem (3-SAT) is to decide whether there exists an assignment
	of truth values to each of the $v$ variables, such that ${\cal F}$ evaluates to true.
	The {\em variable-clause incidence graph} consists of a vertex for each variable and each clause;
	an edge exists between a clause vertex representing a clause $c_j$ and a variable vertex representing a variable
	$x_i$, iff variable $x_i$ occurs in clause $c_j$.
	A 3-SAT instance is {\em planar} if the clause-variable incidence graph is planar.
\end{definition}

It was shown by Lichtenstein~\cite{Lichtenstein82} that 3-SAT is NP-complete, even when restricted to 
planar instances. This is the basis for many NP-completeness reductions for planar graphs and for point sets.
In the following, we give a proof of NP-hardness of the problem CRA with bounded radii in the graph setting. 

\begin{theorem}\label{thm:hardness}
With radii bounded by some constant $\beta$, the problem CRA is NP-hard for planar weighted graphs.
\end{theorem}

\begin{figure}[h!t]
\centering
\includegraphics[width=0.8\columnwidth]{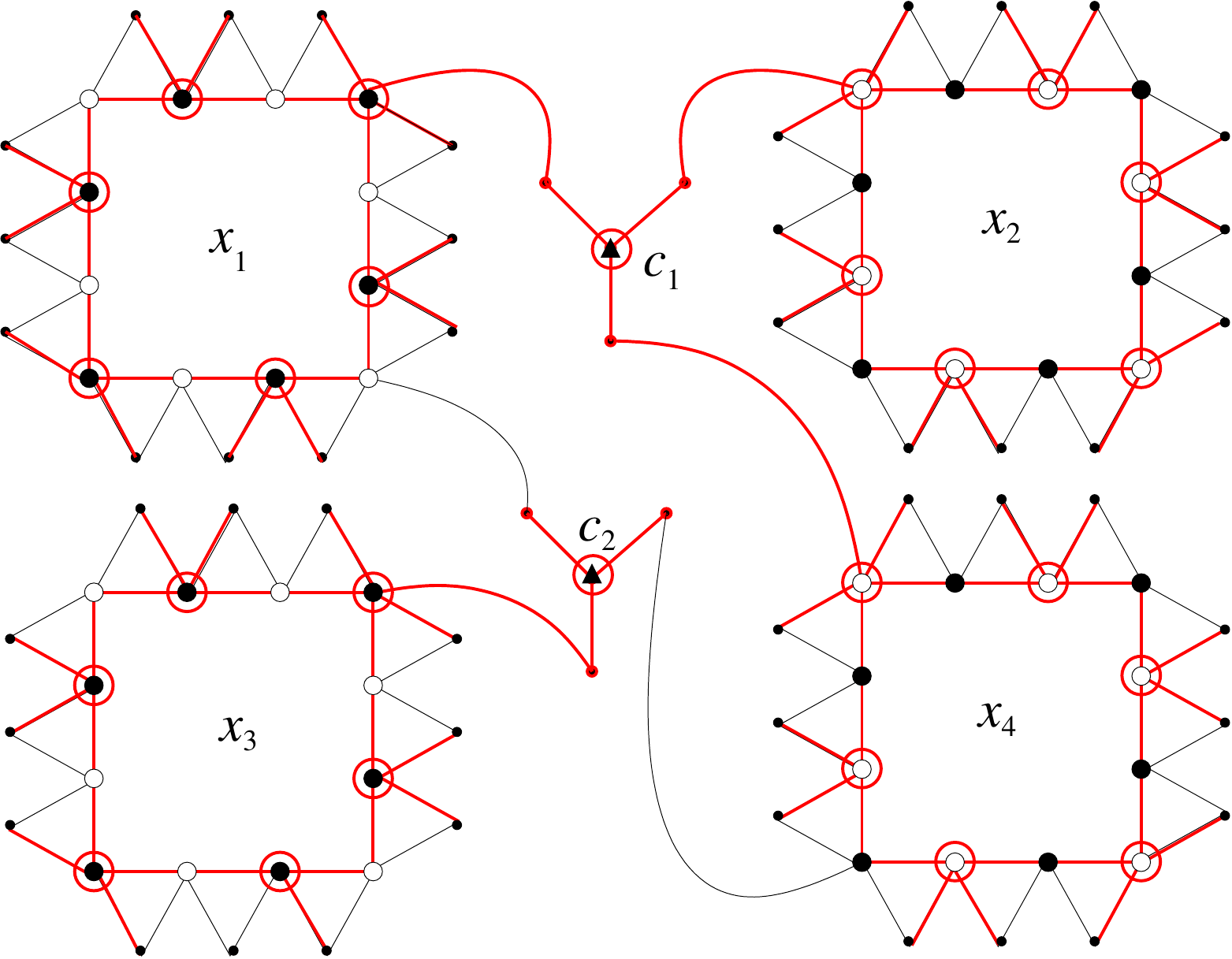}
\caption[Basic gadget construction of the NP-hardness reduction]{
A graph representation for the 3SAT instance 
$(\overline{x}_1\vee {x}_2\vee\overline{x}_3) \wedge
({x}_1\vee \overline{x}_2\vee\overline{x}_4)$.
All displayed edges have length $\beta$. ``True'' and ``false'' 
vertices are marked in bold white or black, respectively; auxiliary vertices are indicated by small dots.
Clause vertices are indicated by small triangles, while the subgraphs for variables are labeled by variable names.
Marked in red are the vertices assigned with a radius of $\beta$, as well as the corresponding connecting edges.
\label{nphard}
}
\end{figure}
\begin{proof}
See Fig.~\ref{nphard} for the main ideas behind the construction.
The proof uses a reduction from {\sc Planar 3Sat}.
Variables are represented by closed ``loops'' at distance $\beta$.
Additional ``connectivity'' edges ensure that
all variable gadgets are connected.
Each clause is represented by a star-shaped set of four points which is connected to each corresponding variable loop over a $\beta$-edge. 
We claim that there is an optimal solution for the constructed CRA instance ``$I_{CRA}(I_{3Sat})$'' with cost $R = v\cdot 6\beta + c\cdot \beta$
if and only if there is a satisfying variable assignment for the corresponding {\sc 3Sat} instance $I_{3Sat}$ with $v$ variables and $c$ clauses.\\

\noindent {\bf ``$\Rightarrow$'':} \quad
	To verify that claim we first state some properties of solutions to an $I_{CRA}(I_{3Sat})$ CRA instance:
	\begin{enumerate}
		\item {\em There are exactly two different feasible connected radii
			assignments of value $6\beta$ for each variable loop}.

			To make this clear see the triangle in Fig.~\ref{fig:np1}
			which illustrates that only ``non-fractional'' solutions provide the required sum of radii.
			More precisely, if more than one disk is used, the sum of radii either exceeds $\beta$
			or one point remains uncovered, as the sum of each two radii must be greater or equal to $\beta$.
			For symmetric reasons this argument yields two valid ``integral'' $6\beta$ solutions for each entire variable loop.
		\begin{figure}[h!]
			\centering
			\includegraphics[width=.8\columnwidth]{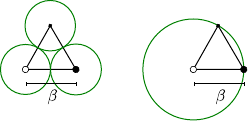}
			\caption[The variable gadget coverage]{The triangular parts of the variable gadgets. A coverage by
			multiple disks (left) always costs more the a coverage with one disk (right).\label{fig:np1}}
		\end{figure}
		\item {\em The four nodes of a clause can be connected to the rest with cost at most $\beta$
			only if there is at least one satisfying variable}.

			Since every clause star center $p_c$ can only be covered by disks around points
			that are not part of the variable loops, the only way to connect the star centers to
			the remainder with cost at most $\beta$ for each clause, is to use exactly one $\beta$-disk 
			centered at one of the four corresponding clause nodes. Using $p_c$ as the $\beta$-disk center
			we must have a satisfying variable loop that closes the gap (of length $\beta$) between the clause
			and the remainder of the graph. Moreover, if we use a different node of the clause as the
			$\beta$-disk center, the remaining two uncovered star nodes must also be covered by disks from the
			variable loops; see Fig.~\ref{fig:np2} for illustration.
		\begin{figure}[h!]
			\centering
		\hspace{1cm}
		\begin{minipage}{0.45\linewidth}
			\includegraphics[width=\linewidth]{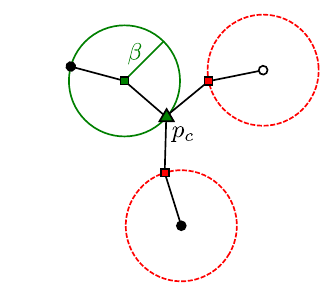}
		\end{minipage}
		\hspace{1cm}
		\begin{minipage}{0.45\linewidth}
			\includegraphics[width=\linewidth]{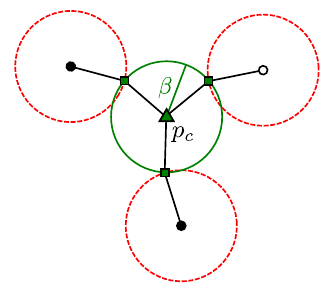}
		\end{minipage}
			\caption[The clause gadget coverage]{We cover each clause star by a $\beta$-disk centered
			    at one of the inner nodes.\label{fig:np2}}
		\end{figure}
	\end{enumerate}
	From the first property we know that the loop construction directly provides a lower bound of $v \cdot 6\beta$
	for the costs of the total variable loop coverage.
	Thus, only costs of $c\cdot \beta$ remain for the coverage of the clauses with total costs $R = v\cdot 6\beta + c\cdot \beta$.
	Combined with the second property we conclude that there is a solution to $I_{CRA}(I_{3Sat})$ with cost
	$R = v\cdot 6\beta + c\cdot \beta$ only if there is a satisfying variable assignment for $I_{3Sat}$.\\

	\noindent{\bf ``$\Leftarrow$'':}\quad
	Given a solution to $I_{3Sat}$. We have seen that any satisfying variable assignment directly encodes
	a variable loop coverage in $I_{CRA}(I_{3Sat})$ with cost $v\cdot(6\beta)$.
	If there is a variable that is not needed (i.e., one for which the value does not impact satisfiability of the clauses),
	it can be either treated as a normal ``true'' or ``false'' variable and is automatically connected to
	another variable loop via the ``connectivity edges''.
	The clause stars can each be covered with a $\beta$-disk centered at the clause star center $p_c$. Thereby, they are all connected to the
	remainder as we have at least one connecting (satisfying) variable loop for each clause. This construction costs $c\cdot \beta$.
	In total, for every solution to $I_{3Sat}$ we get a solution with the desired costs for $I_{CRA}(I_{3Sat})$.
\end{proof}

The main reason why this proof cannot be extended to the case of unbounded radii is that we consider the objective function $Q_1 = \sum r_i$
which does not prefer smaller disks as any $Q_\alpha$ with $\alpha > 1$ does.
That is, although we can construct some local loops that allow only two different integral optimal solutions,
we are unable to see a way to connect these gadgets guaranteeing the absence of an optimal global $1$-disk solution.\\

For the case of bounded radii, we can also show an NP-hardness proof for a geometric setting; again the proof uses
a reduction to planar 3-Satisfiability. For this, we construct variable gadgets that each consist of disk-shaped
point clusters. At the boundary of these clusters we uniformly arrange points such that exactly every third point is used as a disk center in an optimal
solution. Chains of $\beta$-distant points ensure interior connectivity of the variable gadgets.
(This is completely analogous to the construction for the proof of NP-completeness of packing squares into a polygonal domain presented as Theorem~1 in the paper
by Baur and Fekete~\cite{bf-agdp-01}.)
The clause gadgets are constructed as in the graph case. Additional point chains connect them to the corresponding variable loops.

\section{Solutions with a Limited Number of Circles}\label{sec:lim}
A natural class of solutions arises when only a limited number of $k$ disks may have positive
radius. In this section we show that these {\em k-disk solutions} already
yield good approximations; we start by giving a class of lower bounds.

We first observe that the number of disks needed for an optimal solution cannot be bounded by a constant number in general.

\begin{theorem}\label{thm:k_needed}
	For any integer $k \geq 1$, there is an instance of $n = k+3$ points in the plane
	that need exactly $k+1$ disks to be covered optimally.
%	 In addition, the best $k$-disk 
%	solution for this set of points has an approximation ratio at least $(1 + \frac{1}{2^{k+1}-1})$.
\end{theorem}

%\begin{theorem}
%	%\old{Even for a set of collinear points, a}
%	A best $k$-disk solution may be off by a factor of $(1 + \frac{1}{2^{k+1}-1})$.
%\end{theorem}

\begin{figure}[h!]
	\centering
	\includegraphics[width=.8\columnwidth]{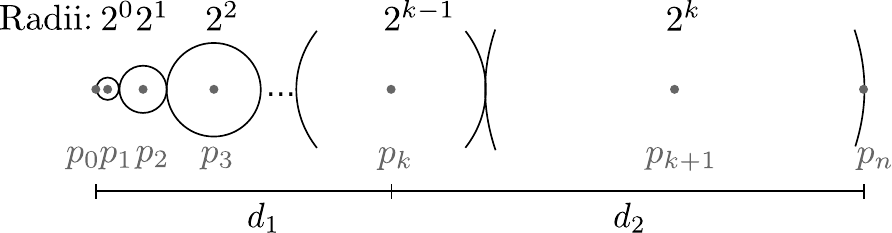}
	\caption{A class of CRA instances that need $k+1$ disks in an optimal solution.\label{k-disk-ex}}
\vspace*{-.2cm}
\end{figure}

\begin{proof}
Consider the example in Fig.~\ref{k-disk-ex} with $n=k+3$ points. The provided solution $R^*$ with $k+1$ disks is optimal,
as $R^* := \sum{r_i^*} = \frac{\dist(p_0,p_n)}{2}$.  To show that the best $k$-disk solution is not optimal,  we make the following observation.
For any integer $k \geq 1$, we have $d_1=d(p_0,p_k) = 2 \cdot \sum_{i=0}^{k-2}{2^i} + 2^{k-1} < 2 \cdot 2^{k} + 2^{k-1} = d(p_k, p_n)= d_2$. 
Therefore, any $k$-disk solution that does not use point $p_{k+1}$ cannot be optimal, as the sum of the radii would be strictly greater than $R^* $.
On the other hand, if the $k$-disk solution contains a disk centered at $p_{k+1}$, it must be of radius $r_{k+1}^*$ in 
order to be optimal. This results in a set of $n-1$ points to be connected optimally using $k-1$ disks. Hence, using induction on $k$, we
conclude that exactly $k+1$ disks are needed.  
%Since we only consider integer
%distances, a best $k$-disk solution has cost $R_k \geq R+1$, i.e.,
%$\frac{R_k}{R} \geq 1 + \frac{1}{2^{k+1}-1}$. \qed
\end{proof}

In the following we establish upper and lower bounds on the approximation ratios of $1$- and $2$-disk solutions for CRA.

\begin{lemma}\label{l2}
Let ${\cal P}$ be a longest (simple) path in a connectivity graph corresponding to an optimal solution $R^*$,
%\old{the connectivity graph of an optimal solution}
 and let $e_m$
 %\old{of length $|e_m|$} 
 be an edge in ${\cal P}$ containing the midpoint of ${\cal P}$.
Then $R^*=\sum r_i^* \geq \max\{\frac{1}{2}|{\cal P}|, |e_m|\}$.
\end{lemma}
This follows directly from the following simple property of the connectivity graph: In this graph,  for any edge $e = p_up_v$ with length $|e|$ in ${\cal P}$, $r_u + r_v \geq |e|=d(p_u,p_v)$.

\begin{theorem}\label{thm:up1}
A best $1$-disk solution for CRA is a $\frac{3}{2}$-approximation, even in the graph version of the problem.
\end{theorem}

\begin{proof}
Consider a longest path ${\cal P} = (p_0,\ldots ,p_k)$ of length
%\old{$|{\cal P}|$}
$|{\cal P}| = d_{\cal P}(p_0,\ldots ,p_k) := \sum_{i=0}^{k-1}{|p_ip_{i+1}|}$ in the connectivity graph of an optimal solution.
Let $R^*~:=~\sum{r_i^*}$ be the cost of the optimal solution, and $e_m~=~p_ip_{i+1}$
as in Lemma \ref{l2}.
Let $\bar d_i := d_{\cal P}(p_i,\ldots ,p_k)$ and $\bar d_{i+1} := d_{\cal P}(p_0,\ldots ,p_{i+1})$. Then
$d_{\min}:=\min\{\bar d_i,\bar d_{i+1}\} \leq \frac{\bar d_i+\bar d_{i+1}}{2}=\frac{d_{\cal P}(p_0,\ldots ,p_i) + 2|e_m| +
d_{\cal P}(p_{i+1},\ldots ,p_k)}{2} = \frac{|{\cal P}|}{2} + \frac{|e_m|}{2} \leq R^* +
\frac{R^*}{2} = \frac{3}{2}R^*$. So a single disk $C$ with a radius of $d_{\min}\leq\frac{3}{2}R^*$ around
the point in $P$ that is nearest to the middle of path ${\cal P}$ covers
${\cal P}$. Without loss of generality, let this center point of $C$ be $p_{i+1}$, then $d_{min}$ = $\bar d_{i+1}$.
We claim that the chosen disk $C$ also covers all points in $P$.
This can be seen by assuming the existence of a point $p' \in P$ that is not
covered by $C$; see Fig.~\ref{fig:up1_2}.
\begin{figure}[h!]
	\centering
	\includegraphics[width=0.6\linewidth]{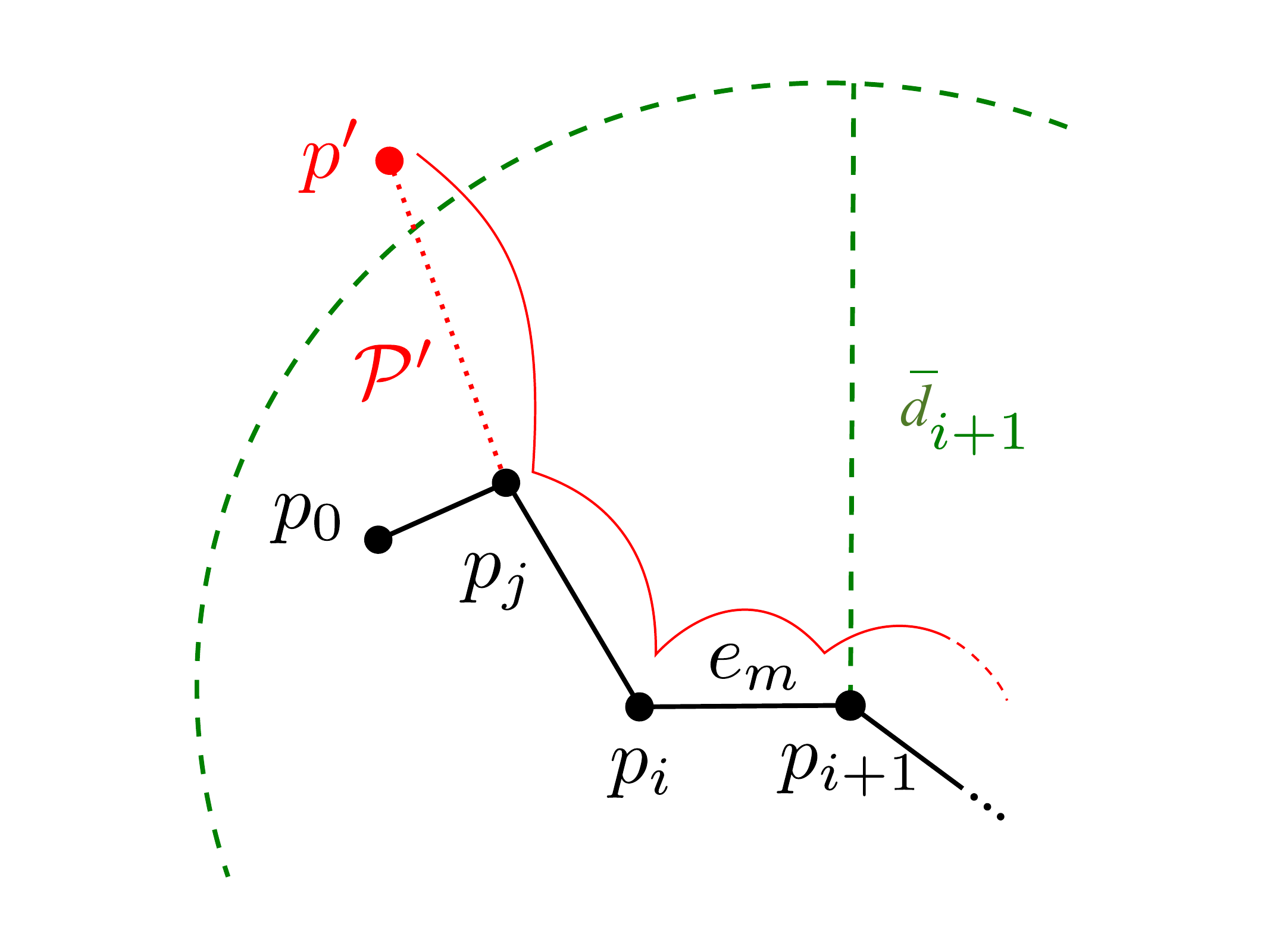}
	\caption[No uncovered point remains from the $1$-disk solution]{The presence of an uncovered
	    point $p'$ (red) implies the existence of a longer path ${\cal P'} \cup (p_j,\ldots,p_k)$ (sketched by the curved red line).}
	\label{fig:up1_2}
\end{figure}
Because the optimal solution must induce a connected connectivity graph, there is a path ${\cal P'}$
from $p'$ to a node $p_j$ in ${\cal P}$ that has no edge in common with ${\cal  P}$, i.e. $V({\cal P'}) \cap V({\cal P}) = \{p_j\}$.
Without loss of generality, let $p_j$ lie between $p_0$ and $p_i$ in ${\cal P}$.
If $p'$ is not covered by $C$, the path
$\cal P'$ must be longer than the path segment $(p_0, \ldots , p_j)$ of ${\cal P}$.
We get a path ${\cal P'} \cup (p_j,\ldots,p_k)$ with length greater than $|{\cal P}|$, which contradicts the maximality of ${\cal P}$.
\end{proof}

\begin{figure}[h!]
	\centering
	\includegraphics[width=0.4\columnwidth]{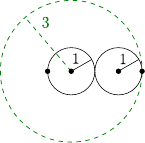}
	\caption{A lower bound of $\frac{3}{2}$ for 1-disk solutions.\label{1_disk_strict}}
\vspace*{-0.2cm}
\end{figure}

Fig.~\ref{1_disk_strict} shows that this bound is tight. There is a 2-disk solution of cost 2 while the best 1-disk solution has cost 3. We now show that using two disks yields an
even better approximation factor. 

\begin{theorem}\label{th:43}
A best $2$-disk solution for CRA is a $\frac{4}{3}$-approximation, even in the graph version
of the problem. 
\end{theorem}

\begin{proof}
Let ${\cal P} = (p_0,\ldots,p_k)$ be a longest path of length
$|{\cal P}| = d_{\cal P}(p_0,\ldots ,p_k) := \sum_{i=0}^{k-1}{|p_ip_{i+1}|}$ in the connectivity graph of an optimal solution with radii $r_i^*$.
Then by Lemma~\ref{l2}, $R^* := \sum r_i^* \geq \frac{1}{2}|{\cal P}|$.
We distinguish two cases;
see Fig.~\ref{fig:43}.

	{\bf Case 1.} There is a point $x$ on ${\cal P}$ at a distance of at
least $\frac{1}{3}|{\cal P}|$ from both endpoints. Then there is a $1$-disk 
solution that is a $\frac{4}{3}$-approximation.
%\old{, and no $2$-disk solution of such quality is needed.}

	{\bf Case 2.} There is no such point $x$. 
	%\old{Then one disk \old{ with radius $s_1$} of a best $2$-disk solution}
Then two disks are needed. One of them is placed
at a point in the first third of ${\cal P}$, and the other disk 
%\old{ with radius $s_2$} 
is placed at a point in the last third of ${\cal P}$.
%\old{ Then $\frac{1}{3}|{\cal P}|\leq s_1+s_2 $.} 
Let $e_m=p_ip_{i+1}$ be
defined as in Lemma~\ref{l2}. Further, let $d_i := d_{\cal P}(p_0,\ldots ,p_i)$,
%\old{be the distance between the path's beginning and $p_i$}
and let $d_{i+1} := d_{\cal P}(p_{i+1},\ldots ,p_k)$.
%\old{be the distance between $p_{i+1}$ and the path's end}
Then $|e_m| = |{\cal P}|-d_i-d_{i+1}$ and $d_i, d_{i+1} < \frac{1}{3}|{\cal P}|$.
%\old{Further we know that $|e_m| > \frac{1}{3}|{\cal P}|$.}

	\begin{figure}[h!]
		\centerline{\includegraphics[width=0.8\linewidth]{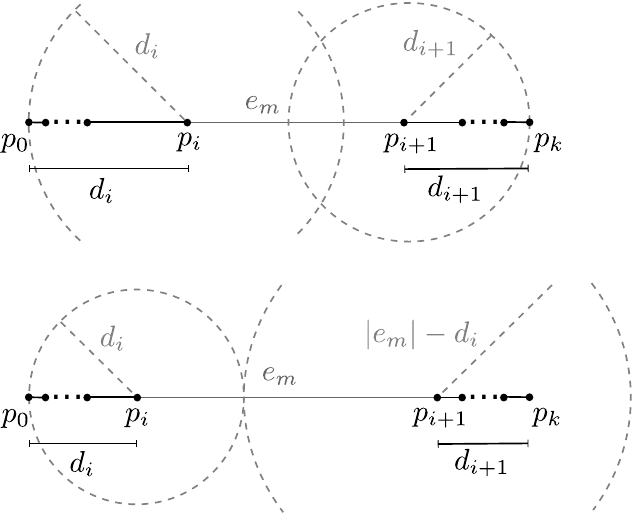}}
		\caption[A best 2-disk solution is always a $\frac{4}{3}$-approximate solution.]
			{Two $\frac{4}{3}$-approximate $2$-disk solutions (light grey, dashed):
			(Top) Case 2a, when $|e_m| < d_i + d_{i+1}$;
			(bottom) Case 2b, where $|e_m| > d_i + d_{i+1}$.}
                \label{fig:43}
	\end{figure}
	
	{\bf Case 2a.} If $|e_m| < \frac{1}{2}|{\cal P}|$
then $d_i + d_{i+1} = |{\cal P}| - |e_m| > \frac{1}{2}|{\cal P}| > |e_m|$.
Set $r_i := d_i$ and $r_{i+1} := d_{i+1}$, then the path is covered.
Analogously to the proof of Theorem~\ref{thm:up1}, we conclude that there is no
point $p' \in P$ uncovered by the chosen disks, as there will otherwise be a
longer path than ${\cal P}$.
Since $d_i, d_{i+1} < \frac{1}{3}|{\cal P}|$ we have
$r_i + r_{i+1} = d_i + d_{i+1} < \frac{2}{3}|{\cal P}| \leq \frac{4}{3}R^*
%\old{\sum r^*_i}
$
and the claim holds.
	
	{\bf Case 2b.} Otherwise, if $|e_m| \geq \frac{1}{2}|{\cal P}|$
then $d_i + d_{i+1} %\old{= |{\cal P}| - |e_m|} 
\leq \frac{1}{2}|{\cal P}| \leq |e_m|$. 
Assume $d_i \geq d_{i+1}$. Choose $r_i := d_i$ and $r_{i+1} := |e_m|-d_i$.
As $d_{i+1} \leq |e_m| - d_i$ the path ${\cal P}$ is covered. Again the existence of an uncovered point contradicts the maximality of ${\cal P}$. Finally,
 $r_i + r_{i+1} = d_i + (|e_m| - d_i) = |e_m|$,  which is
the lower bound and thus the range assignment is optimal. 
\end{proof}

If all points of $P$ lie on a straight line, the approximation ratio
for two disks can be improved. We first observe an interesting property of optimal solutions in this case.

\begin{lemma} \label{overl_lem}
	Let $P$ be a set of points on a straight line. Then there is a non-overlapping
optimal solution, i.e., one in which all disks have disjoint interior.
\end{lemma}
\begin{proof}
	An arbitrary optimal solution is modified as follows.
For every
two overlapping disks $C_i$ and $C_{i+1}$ with centers $p_i$ and $p_{i+1}$, we
decrease $r_{i+1}$, such that $r_i + r_{i+1} = \dist(p_i,p_{i+1})$, and increase
the radius of $C_{i+2}$ by the same amount.
This can be iterated, until there is at most one overlap at the outermost disk $C_j$ (with $C_{j-1}$).
Then there must be a point $p_{j+1}$ on the boundary of $C_j$: otherwise we could shrink $C_j$ contradicting
optimality. Decreasing $C_j$'s radius $r_{j}$ by the overlap $l$ and
adding a new disk with radius $l$ around $p_{j+1}$ creates an optimal
solution without overlap. 
\end{proof}

\begin{theorem}
\label{th:54}
	Let $P$ be a set of points on a straight line $g$. Then, a best $2$-disk 
        solution for CRA is a $\frac{5}{4}$-approximation.
\end{theorem}

\begin{proof}
	\begin{figure}[h!]
		\centering
		\includegraphics[width=0.8\linewidth]{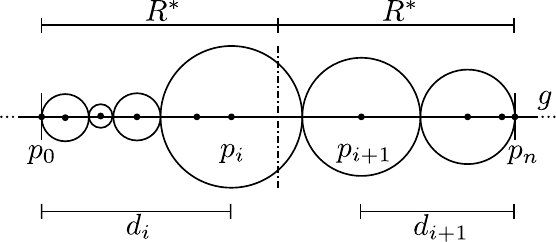}
		\caption{A non-overlapping optimal solution.}
		\label{fig:line_notation}
	\end{figure}
According to Lemma \ref{overl_lem}, we are, without loss of generality, given an optimal
solution $R^*$ with non-overlapping disks. Let $p_0$ and $p_n$ be the outermost
intersection points of the optimal solution disks and $g$. Without loss of generality,
we may further assume $p_0,p_n \in P$ and $R^* := \sum{r_i^*} =
\frac{\dist(p_0,p_n)}{2}$ (otherwise, we can add the outermost intersection
point of the outermost disk and $g$ to $P$, which may only improve
the approximation ratio).
Let $p_i$ denote the rightmost point in $P$ left to the middle of $\overline{p_0 p_n}$ and let $p_{i+1}$ its neighbor on the other half.
Further, let $d_i := \dist(p_0,p_i)$, $d_{i+1} := \dist(p_{i+1},p_n)$ (See Fig.~\ref{fig:line_notation}). Assume, $d_i \geq
d_{i+1}$. We now give $\frac{5}{4}$-approximate solutions using one or two
disks that cover $\overline{p_0 p_n}$.
	
	{\bf Case 1.} If $\frac{3}{4} R^* \leq d_i$ then $\frac{5}{4} R^* \geq 2R^* - d_i = \dist(p_i,p_n)$.
Thus, the solution consisting of exactly one disk with radius $2R^* - d_i$ centered at $p_i$ is sufficient.
	
	{\bf Case 2.} If $\frac{3}{4} R^* > d_i \geq d_{i+1}$ we need two disks 
to cover $\overline{p_0 p_n}$ with $\frac{5}{4}R^*$.
	
	\begin{figure}[h!]
		\centering
		\includegraphics[width=0.8\linewidth]{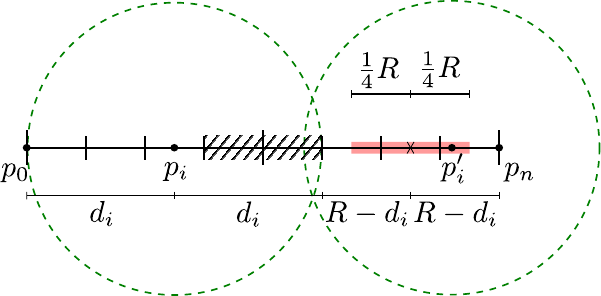}
		\caption{A $\frac{5}{4}$-approximate $2$-disk solution with $d_i < \frac{3}{4}R^*$. The cross marks the position of the optimal counterpart $p_i^*$ to $p_i$ and the grey area sketches $A_i$.}    %%% should we add the label $p_i^*$ to the ``cross''?  Why not label the point?   yyy
		\label{fig:line_solution}
	\end{figure}
	
	{\bf Case 2a.}The point $p_i$ could be a center point of an optimal two-disk solution if there was a point $p_i^*$ with $\dist(C_i,p_i^*) = \dist(p_i^*,p_n) = R^* - d_i$.
So in case there is a $p_i' \in P$ that lies in a $\frac{1}{4}R^*$-neighborhood of such an optimal $p_i^*$ we get $\dist(C_i,p_i'), \dist(p_i',p_n) \leq R^*-d_i+\frac{1}{4}R^*$ (see Fig.~\ref{fig:line_solution}). Thus, $r(p_i) := d_i, r(p_i'):= R^*-d_i+\frac{1}{4}R^*$ provides a $\frac{5}{4}$-approximate solution.
	
	{\bf Case 2b.} Analogously to Case 2a, there is a point $p_{i+1}' \in P$ within
a $\frac{1}{4}R^*$-range of an optimal counterpart to $p_{i+1}$.
Then we can take $r(p_{i+1}):=d_{i+1}$, $r(p_{i+1}') := R^*-d_{i+1}+\frac{1}{4}R^*$ as a
$\frac{5}{4}$-approximate solution.

	{\bf Case 2c.} Assume that there is neither such a $p_i'$ nor such a
$p_{i+1}'$. Since $d_i, d_{i+1}$ are in $(\frac{1}{4}R^*,\frac{3}{4}R^*)$, we have 
$\frac{1}{4}R^* < R^* - d_j < \frac{3}{4}R^*$ for $j = i,i+1$, which implies that there
are two disjoint areas $A_i$, $A_{i+1}$, each with diameter equal to
$\frac{1}{2}R^*$ and excluding all points of $P$. Since $p_i$, the rightmost point
on the left half of $\overline{p_0 p_n}$, has a greater distance to $A_i$ than to
$p_0$, any disk around a point on the left could only cover parts of both
$A_i$ and $A_{i+1}$ if it has a greater radius than its distance to $p_0$. This
contradicts the assumption that $p_0$ is a leftmost point of a disk in an optimal
solution. The
same applies to the right-hand side. Thus, $A_i \cup A_{i+1}$ must contain at
least one point of $P$, and therefore one of the previous cases leads to a
$\frac{5}{4}$-approximation. 
\end{proof}

\begin{figure}[h!]
	\centering
	\includegraphics[width=0.8\linewidth]{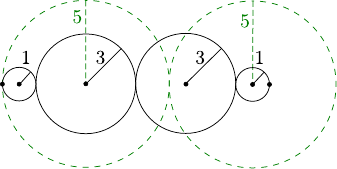}
	\caption{A lower bound of $\frac{5}{4}$ for 2-disk solutions.\label{2_disk_strict}}
\end{figure}

Fig.~\ref{2_disk_strict} shows that the bound is tight. We believe that this
is also the worst case when points are {\em not} on a line. Indeed,
the solutions constructed in the proof of Theorem~\ref{th:54}
cover a longest path ${\cal P}$ in an optimal
solution for a general $P$.  If this longest path consists of at most three
edges, $p_i (=:p_{i+1}')$ and $p_{i+1} (=:p_i')$ can be chosen as disk 
centers, covering all of $P$. However, if ${\cal P}$ consists of at least
four edges, a solution for the diameter may produce two internal non-adjacent
center points that do not necessarily cover all of $P$.

\section{A Polynomial-Time Approximation Scheme for CRA}\label{sec:ptas}

%% reviewer comment (CGTA):
\old{ xxx
Section 5 was poorly written. It is arguably the most technically demanding part of the paper, and there
is not a single diagram. Additionally they do things like use a concept (like m-gap) in a proof, and then try
to define it on the fly after it was used. I had to read it multiple times before I understood what they
were saying. There are many concepts that are poorly defined, or defined in fragments here and there.
They also make a claim that they attempt to prove, again ``on the fly'' within another proof, but they do
not make any kind of convincing argument. At this point for all I know it is not true, and the claim is false.

Sections 5.1.3 and 5.2 have concepts that are defined ambiguously or incompletely. In addition these
sections rely on the clarity and correctness of the preceding section 5.1.2, which is presently lacking.

Joe: I don't see where ``m-gap'' is used at all??  I will have to decipher what the referee meant in these remarks...
I think by ``m-gap'' they mean ``m-span''?
}

We have seen (Theorem~\ref{thm:hardness}, Section~\ref{sec:np-proof}) that the problem CRA is NP-hard when there are upper bounds given on the possible radii, $r_i$, of the input points $p_i\in P$, even if these upper bounds are all the same value $\beta$.
We do not yet know the complexity of exactly solving the problem CRA
in the case in which the radii $r_i$ are not bounded above.  In this
section, we give a polynomial-time approximation scheme (PTAS) for CRA
for this case ($\beta=\infty$): for any fixed $\epsilon>0$, we
compute, in polynomial time, a set of radii, $r_i$, so that the disks
of radii $r_i$ centered at the respective input points $p_i$ form a
connected set of disks, and the sum of radii $\sum_i r_i$ is at most
$(1+\epsilon)$ times optimal.
% Its complexity is not known otherwise.
% We now design polynomial-time approximation schemes (PTAS) for both of these cases.

% \subsection{A PTAS for CRA}

%We begin with the case in which $\bar r_i=\infty$, for each $i$.  
To design a PTAS for this problem, we first prove a structure theorem that allows us
to apply the $m$-guillotine method (\cite{m-gsaps-99}) and transform any solution into an {\em ``$m$-guillotine solution''} (defined precisely below) with a small increase in the sum of radii. This theorem
permits the subsequent use of dynamic programming to obtain a PTAS. 

\subsection{Structural Results}

For our first structural result, we need some definitions.  Consider a
set ${\cal D}$ of disks, centered at the points
$P=\{p_1,p_2,\ldots,p_n\}$.  Let $\delta=diam(P)$ denote the diameter
of the point set; i.e., $\delta$ is the maximum of the (Euclidean)
distances between pairs of points of $P$.  If the radii of the disks
${\cal D}$ are all from the discrete set ${\cal R}=\{\delta/mn,
2\delta/mn, \ldots, \delta\}$, then we will say that ${\cal D}$ is a
set of {\em ${\cal R}_{m,P}$-disks}, or {\em ${\cal R}$-disks}, for
short, with the understanding that $m=O(1/\epsilon)$ and $P$ will be
fixed throughout our discussion.  Now consider a set ${\cal D}$ of
${\cal R}$-disks, centered on the points $P$.  We let ${\cal I}_x$
(resp., ${\cal I}_y$) denote the set of $x$-coordinates (resp.,
$y$-coordinates) of the set of all coordinate-extreme points
(leftmost, rightmost, topmost, and bottommost) of the disks ${\cal
  D}$.  Specifically, ${\cal I}_x=\{x(p_i) \pm j(\delta/mn): 1\leq
i\leq n, 0\leq j\leq mn\}$ and ${\cal I}_y=\{y(p_i) \pm j(\delta/mn):
1\leq i\leq n, 0\leq j\leq mn\}$, where $x(\cdot)$ and $y(\cdot)$
denote $x$- and $y$-coordinates.

We begin with a simple lemma that shows that we can round up the radii
of a feasible solution, to make it a set of ${\cal R}$-disks, at a
small cost to the objective function:

\begin{lemma}
\label{lem:round}
Let $\sum_i r_i$ be the sum of radii in a feasible (connected union)
solution, ${\cal D}$.  Then, for any fixed $\epsilon>0$, there exists
a set, ${\cal D}_m$, of $n$ disks of radii $r'_i$ centered on points
$p_i$, such that (a). $r'_i\in {\cal R}=\{\delta/mn,
2\delta/mn,\ldots, \delta\}$, where $\delta=diam(P)$ is the diameter
of the input point set $P$ and $m=\lceil 2/\epsilon\rceil$; and
(b). $\sum_i r'_i \leq (1+\epsilon)\sum_i r_i$.
\end{lemma}

\begin{proof}
Each of the $n$ radii $r_i$ can be increased by at most $\delta/mn\leq
\epsilon \delta/2n$ at a total cost (increase in the sum of all radii)
of at most $\epsilon \delta/2$.  Since increasing the radii of the
disks keeps the set of disks connected, and since $\sum_i r_i\geq
\delta/2$, we obtain the result.
\end{proof}

Next, we state and prove a simple observation about the structure of
an optimal solution:

\begin{lemma}
\label{lem:6}
Let ${\cal D}$ be a set of circular disks of radii $r_i$ centered at
points $p_i\in P$, such that the union of the disks is connected.
Then, there exists a set ${\cal D}'$ of circular disks centered at
points $p_i\in P$, such that the union of the disks ${\cal D}'$ is
connected, the sum of the radii of disks ${\cal D}'$ is at most
$\sum_i r_i$, and no point in the plane lies within more than 6 disks
of ${\cal D}'$.  In particular, in an optimal solution, no point in
the plane lies within more than 6 disks of the set of optimal disks.  Further, if
${\cal D}$ are ${\cal R}$-disks, then there exists a set ${\cal D}'$
of ${\cal R}$-disks of total radii $\sum_i r_i$ having the property
that no point in the plane lies within more than 6 disks of ${\cal
  D}'$.
\end{lemma}

\begin{proof}
If some point $p$ lies within more than 6 disks of ${\cal D}$, then we
know that two of them, say $C_1$ and $C_2$, have centers, $p_1$ and
$p_2$, within a 60-degree cone with apex $p$.  Thus, the distance
$|p_1p_2|$ between the two centers is less than the larger, say $r_1$
(the radius of $C_1$), of the two radii ($r_1$ and $r_2\leq r_1$),
implying that center $p_2$ lies within the disk $C_1$.  Thus, if we
enlarge the disk $C_1$ to have radius $r_2 + |p_1p_2| \leq r_1+r_2$,
the enlarged disk completely covers disk $C_2$, allowing $C_2$ to be
deleted (i.e., allowing the radius of $C_2$ to be shrunk to 0).  This
modification yields a set ${\cal D}'$ of disks that are connected
(since the union of the disks of ${\cal D}'$ contains the union of the
disks of ${\cal D}$), whose total sum of radii is at most that of the
set ${\cal D}$.

If the disks ${\cal D}$ are ${\cal R}$-disks, then the same argument
applies, but we now replace the radius $r_1$ with the enlarged radius
$r_1+r_2$ ($\geq r_2+|p_1p_2|$), while shrinking $r_2$ to zero.  Since
the disks ${\cal D}$ are ${\cal R}$-disks, we know that $r_1, r_2 \in
{\cal R}=\{\delta/mn, 2\delta/mn, \ldots, \delta\}$, and therefore
$r_1+r_2\in {\cal R}$.  Thus, disks ${\cal D}'$ are ${\cal R}$-disks.
\end{proof}

Our main structural result requires further definitions.
Let ${\cal D}$ be a set of $n$ disks; we consider the disks to be {\em closed} (include their bounding circles).
Let $\rho$ be an axis-parallel rectangle.
An axis-parallel line $\ell$ that intersects $\rho$ is said to be a {\em cut} of $\rho$.
A cut $\ell$ is said to be {\em $m$-perfect} with respect to ${\cal D}$ and the rectangle $\rho$
if $\ell$ intersects the interiors of at most $c_1 m + c_2$ % $2m$     xxx do I want to stick with 2m ??
disks of ${\cal D}$ that 
lie interior to~$\rho$, for fixed constants (independent of $m, \epsilon$) $c_1$ and $c_2$.
% have a nonempty intersection with $\rho$.

We say that a set ${\cal D}$ of ${\cal R}$-disks is {\em
  $m$-guillotine with respect to (axis-aligned) rectangle $\rho$} if
either (1) no disk of ${\cal D}$ lies interior to $\rho$, or (2) there
exists an (axis-parallel) cut $\ell$, defined by coordinates ${\cal
  I}_x$ and ${\cal I}_y$, that is $m$-perfect with respect to ${\cal
  D}$ and $\rho$, such that ${\cal D}$ is $m$-guillotine with respect
to both of the rectangles into which $\rho$ is partitioned by the cut
$\ell$.  We say that a set ${\cal D}$ of ${\cal R}$-disks is {\em
  $m$-guillotine} if ${\cal D}$ is $m$-guillotine with respect to the
bounding box, $BB({\cal D})$, of ${\cal D}$.  See
Fig.~\ref{fig:example} for an example of a set of disks that is
1-guillotine.

\begin {figure}[htbp!]%[tb]
\centering
\begin{picture}(0,0)%
\includegraphics{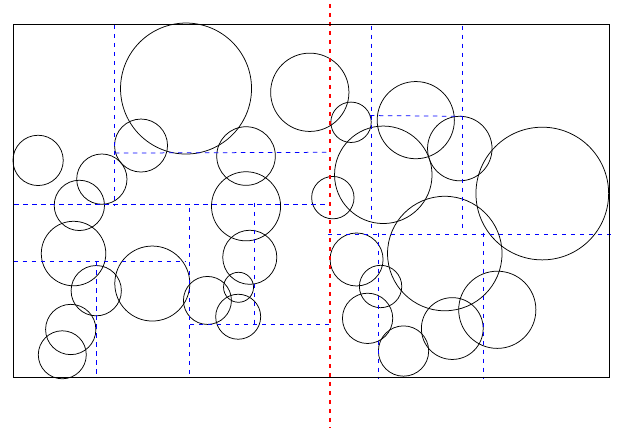}%
\end{picture}%
\setlength{\unitlength}{1579sp}%
\begingroup\makeatletter\ifx\SetFigFont\undefined%
\gdef\SetFigFont#1#2#3#4#5{%
  \reset@font\fontsize{#1}{#2pt}%
  \fontfamily{#3}\fontseries{#4}\fontshape{#5}%
  \selectfont}%
\fi\endgroup%
\begin{picture}(12514,8553)(406,-7924)
\put(436,-4681){\makebox(0,0)[lb]{\smash{{\SetFigFont{8}{9.6}{\familydefault}{\mddefault}{\updefault}{\color[rgb]{0,0,0}4}%
}}}}
\put(7786,-6646){\makebox(0,0)[lb]{\smash{{\SetFigFont{8}{9.6}{\familydefault}{\mddefault}{\updefault}{\color[rgb]{0,0,0}4}%
}}}}
\put(5386,-2311){\makebox(0,0)[lb]{\smash{{\SetFigFont{8}{9.6}{\familydefault}{\mddefault}{\updefault}{\color[rgb]{0,0,0}4}%
}}}}
\put(4156,-7291){\makebox(0,0)[lb]{\smash{{\SetFigFont{8}{9.6}{\familydefault}{\mddefault}{\updefault}{\color[rgb]{0,0,0}3}%
}}}}
\put(7636,254){\makebox(0,0)[lb]{\smash{{\SetFigFont{8}{9.6}{\familydefault}{\mddefault}{\updefault}{\color[rgb]{0,0,0}4}%
}}}}
\put(9676,269){\makebox(0,0)[lb]{\smash{{\SetFigFont{8}{9.6}{\familydefault}{\mddefault}{\updefault}{\color[rgb]{0,0,0}3}%
}}}}
\put(8641,-1636){\makebox(0,0)[lb]{\smash{{\SetFigFont{8}{9.6}{\familydefault}{\mddefault}{\updefault}{\color[rgb]{0,0,0}5}%
}}}}
\put(12736,-4096){\makebox(0,0)[lb]{\smash{{\SetFigFont{8}{9.6}{\familydefault}{\mddefault}{\updefault}{\color[rgb]{0,0,0}2}%
}}}}
\put(2206,-7231){\makebox(0,0)[lb]{\smash{{\SetFigFont{8}{9.6}{\familydefault}{\mddefault}{\updefault}{\color[rgb]{0,0,0}5}%
}}}}
\put(6736,389){\makebox(0,0)[lb]{\smash{{\SetFigFont{8}{9.6}{\familydefault}{\mddefault}{\updefault}{\color[rgb]{0,0,0}1}%
}}}}
\put(5956,-5956){\makebox(0,0)[lb]{\smash{{\SetFigFont{8}{9.6}{\familydefault}{\mddefault}{\updefault}{\color[rgb]{0,0,0}4}%
}}}}
\put(5266,-4591){\makebox(0,0)[lb]{\smash{{\SetFigFont{8}{9.6}{\familydefault}{\mddefault}{\updefault}{\color[rgb]{0,0,0}5}%
}}}}
\put(2581,269){\makebox(0,0)[lb]{\smash{{\SetFigFont{8}{9.6}{\familydefault}{\mddefault}{\updefault}{\color[rgb]{0,0,0}3}%
}}}}
\put(421,-3631){\makebox(0,0)[lb]{\smash{{\SetFigFont{8}{9.6}{\familydefault}{\mddefault}{\updefault}{\color[rgb]{0,0,0}2}%
}}}}
\put(10051,-7186){\makebox(0,0)[lb]{\smash{{\SetFigFont{8}{9.6}{\familydefault}{\mddefault}{\updefault}{\color[rgb]{0,0,0}$3$}%
}}}}
\put(11551,-6661){\makebox(0,0)[lb]{\smash{{\SetFigFont{12}{14.4}{\familydefault}{\mddefault}{\updefault}$\rho$}}}}
\put(7066,-7726){\makebox(0,0)[lb]{\smash{{\SetFigFont{12}{14.4}{\familydefault}{\mddefault}{\updefault}$\ell$}}}}
\end{picture}%

%%%
\caption{
    An example of an $m$-guillotine set of disks. The vertical red cut $\ell$ is the first cut (labelled ``1'') in the recursive decomposition of the rectangle $\rho$.  The blue cuts occur in the recursive decomposition; each is labeled with an integer (``2'' through ``5'') indicating the level in the recursive decomposition. Each cut intersects the interior of at most 2 disks that are interior to the rectangle corresponding to the cut. All cuts are at coordinates ${\cal I}_x$ or ${\cal I}_y$. In the final decomposition of $\rho$, no disk lies interior to a rectangular face (some do lie within a face, in contact with the boundary).  Thus, the set of disks is 1-guillotine, with constants $c_1=2$, $c_2=0$. \label{fig:example}}
\end{figure}
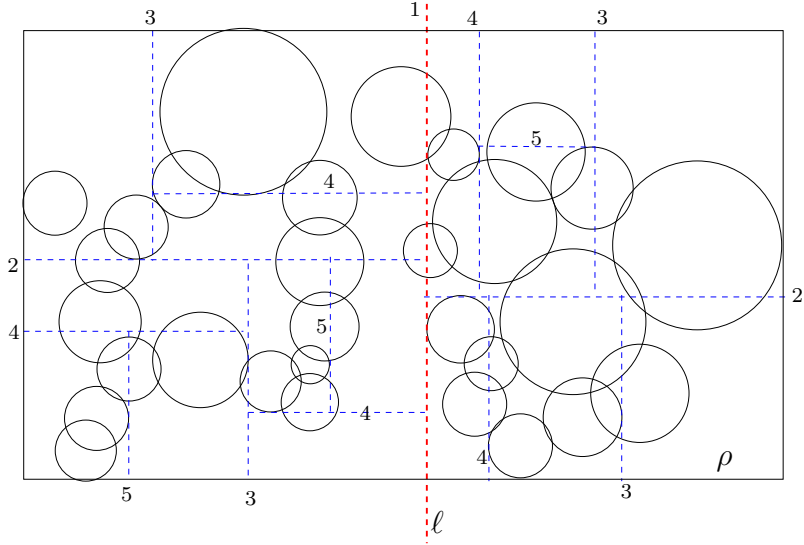

% \subsection{The Structure Theorem}
We are now ready to state and prove our structure theorem, which shows that we can
transform an arbitrary set ${\cal D}$ of circular disks centered on
points $P$, having a connected union and a sum of radii $R=\sum_i r_i$, into an
$m$-guillotine set of ${\cal R}$-disks, ${\cal D}_m$, having sum of
radii at most $(1+\epsilon)R$ and having a connected union.  More
specifically, we show:

\begin{theorem}
\label{thm:structure}
Let ${\cal D}$ be a set of $n$ circular disks of radii $r_i$ centered at points
$p_i\in P$, such that the union of the disks is connected.  Then, for
any fixed $\epsilon>0$, there exists an $m$-guillotine set ${\cal
D}_m$ of at most $n$ ${\cal R}$-disks such that the union of the circular disks ${\cal
D}_m$ is connected and the sum of the radii of disks of ${\cal D}_m$
is at most $(1+(C/m))\sum_i r_i$. Here, $m=\lceil 2/\epsilon \rceil$
and $C$ is a constant.
\end{theorem}

\begin{proof}
First, by Lemma~\ref{lem:6}, we can assume that no point in the plane
lies in more than 6 of the input disks.

Consider the set of bounding boxes (squares) of the disks ${\cal D}$.
For each disk of radius $r_i$, we consider the 4 edges of its bounding
square, together with the vertical and horizontal diameter segments of
the disk.  For each disk, the resulting arrangement of
horizontal/vertical line segments forms an arrangement (a square with
a ``$+$'' inside it, partitioning the square into four subsquares) of
12 line segments, each of length $r_i$.  We let $E$ denote this set of
$12n$ line axis-parallel line segments.  Let $BB(E)$ denote the
axis-aligned bounding box of~$E$, which is the same as $BB({\cal D})$,
the bounding box of the disks.  The edges $E$ have their endpoint
coordinates among ${\cal I}_x$ and ${\cal I}_y$, and they form a
network of horizontal/vertical segments within $BB(E)$, of total
length $\lambda(E)=12\sum_i r_i$.

We utilize the concept of an ``$m$-span'' from \cite{m-gsaps-99}.  
\old{Let
$E$ be a given finite set of {\em edges}, each of which is a line
segment or a curved arc, in the plane.}
\old{We assume that each curved arc is algebraic, as well as $x$-monotone
and $y$-monotone, meaning that the intersection of any arc with any
horizontal or vertical line is either empty or a single line segment.
(The definitions can be extended to more general classes of curves,
but this is not needed here, since our curved arcs are quarter
circles.)}
Let $\rho$ be an axis-aligned rectangle. The intersection, $\ell\cap
(E\cap \rho)$, of a horizontal/vertical cut line $\ell$ with $E\cap
\rho$ consists of a discrete (possibly empty) set of line segments and
singleton intersection points (which we can consider to be zero-length line segments).
Let the endpoints of these segments, be denoted by $q_1,\ldots,q_{\xi}$, in order along~$\ell$.
We define the {\em $m$-span}, $\sigma_m(\ell)$, of $\ell$ (with
respect to $\rho$) to be the empty set if $\xi\leq 2(m-1)$, and to be
the (possibly zero-length) line segment $q_{m}q_{\xi-m+1}$ otherwise.
Refer to Fig.~\ref{fig:m-span-E} for an example.

\old{
Consider the set of bounding boxes (squares) of the disks ${\cal D}$;
they have coordinates among ${\cal I}_x$ and ${\cal I}_y$.  Let $E$ be
the set of boundary edges of these squares.  The edges $E$ form a
network of horizontal/vertical segments within $\rho$, of total length
$\lambda(E)$.  }
\old{Each of the $n$ disks in ${\cal D}$ is bounded by a circle, which we
view as a union of four quarter-circular arcs, having endpoints at the
topmost/bottommost and leftmost/rightmost points of each disk.  Let
$E$ be the set of these $4n$ arcs.  The total length of the arcs in
$E$ is $\lambda(E)=\sum_i 2\pi r_i$, where $r_i$ is the radius of the $i$th disk
of ${\cal D}$.  Let $BB(E)$ denote the axis-aligned bounding box of~$E$.}

By the $m$-guillotine charging argument~\cite{m-gsaps-99}, we know
that the edge set $E$ can be augmented by a set of $m$-span
horizontal/vertical segments (bridges), of total length
$O(\lambda(E)/m)=O((1/m)\sum_i r_i)$, so that the resulting edge set
$E'$ 
%
% (which now consists of quarter-circular arcs and horizontal/vertical line segments) 
%
is $m$-guillotine with respect to
$BB(E)$ in the usual sense defined in~\cite{m-gsaps-99}: there exists
a horizontal/vertical cut
%
% (through coordinates ${\cal I}_s$, ${\cal I}_y$) 
%
for $BB(E)$ such that the $m$-span of $E'$ is contained in $E'$, and,
recursively, the set $E'$ is $m$-guillotine with respect to the
subrectangles of $\rho$ on either side of the cut.  The base case of
the recursion is specified by defining $E'$ to be $m$-guillotine with
respect to rectangle $\rho$ if no vertex (endpoint) of an edge of $E'$
lies interior to $\rho$.
%
%% xxx refer to a figure?

\begin {figure}[htbp!]%[tb]
\centering
\begin{picture}(0,0)%
\includegraphics{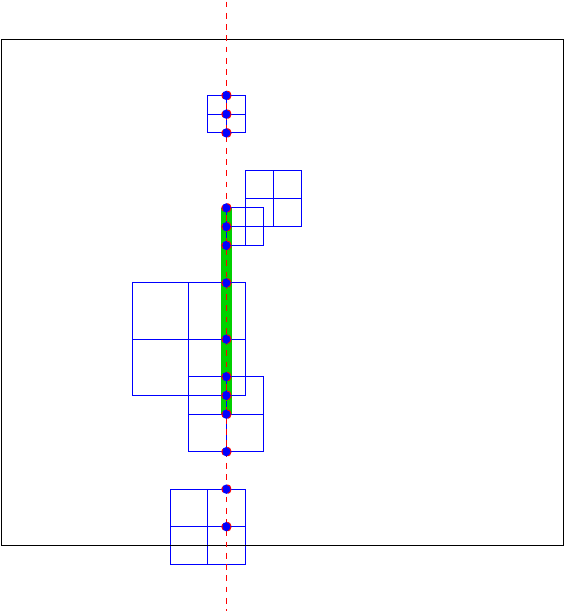}%
\end{picture}%
\setlength{\unitlength}{1973sp}%
\begingroup\makeatletter\ifx\SetFigFont\undefined%
\gdef\SetFigFont#1#2#3#4#5{%
  \reset@font\fontsize{#1}{#2pt}%
  \fontfamily{#3}\fontseries{#4}\fontshape{#5}%
  \selectfont}%
\fi\endgroup%
\begin{picture}(9044,9794)(1179,-9533)
\put(4951,-9286){\makebox(0,0)[lb]{\smash{{\SetFigFont{12}{14.4}{\familydefault}{\mddefault}{\updefault}$\ell$}}}}
\put(4201,-3136){\makebox(0,0)[lb]{\smash{{\SetFigFont{10}{12.0}{\familydefault}{\mddefault}{\updefault}$q_4$}}}}
\put(8851,-8011){\makebox(0,0)[lb]{\smash{{\SetFigFont{12}{14.4}{\familydefault}{\mddefault}{\updefault}$\rho$}}}}
\put(5251,-1336){\makebox(0,0)[lb]{\smash{{\SetFigFont{10}{12.0}{\familydefault}{\mddefault}{\updefault}$q_1$}}}}
\put(5176,-8236){\makebox(0,0)[lb]{\smash{{\SetFigFont{10}{12.0}{\familydefault}{\mddefault}{\updefault}$q_\xi=q_{14}$}}}}
\put(5251,-1936){\makebox(0,0)[lb]{\smash{{\SetFigFont{10}{12.0}{\familydefault}{\mddefault}{\updefault}$q_3$}}}}
\put(5251,-1636){\makebox(0,0)[lb]{\smash{{\SetFigFont{10}{12.0}{\familydefault}{\mddefault}{\updefault}$q_2$}}}}
\end{picture}%

%%%
\caption{
    An example of an $m$-span (bold green) for a set of edges $E$, shown in blue.  Here, $m=4$, and the cut $\ell$ has $\xi=14$ endpoints of segments of $E$ along it. 
\label{fig:m-span-E}}
\end{figure}
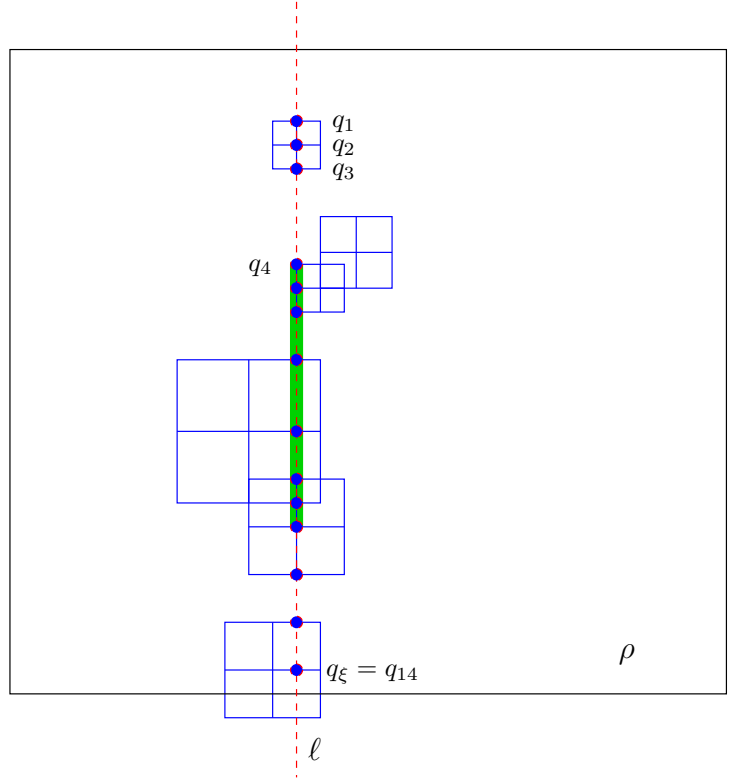

%%%%%
\old{
We say that segment $ab$ {\em fully crosses} a square $\Sigma$ if
$ab\cap \Sigma = \ell\cap \Sigma$; i.e., $ab$ crosses two opposite
sides of $\Sigma$.  We then distinguish two cases:}
%%%%%%
We say that an axis-parallel segment $ab\subset \ell$ {\em fully
  crosses} a disk $D$ if $ab$ intersects the interior of $D$ and
$ab\cap D = \ell\cap D$; i.e., $ab$ crosses two opposite (top/bottom
or left/right) semicircles that comprise the boundary of $D$, so that
as one traverses the segment $ab$, one enters $D$ exactly once and
then exits $D$ exactly once.  Note that if $ab\subset \ell$ fully
crosses $D$, then the projection of the center of $D$ onto $\ell$ lies
on the segment $ab$.
Note too that if segment $ab$ fully crosses a disk $D$, then $ab$
properly crosses the (horizontal/vertical) diameter segment of $D$
that is perpendicular to~$ab$

If $E'=E$ (i.e., no $m$-span segments had to be added to $E$, as $E$
is already $m$-guillotine as an edge set), then the disks ${\cal D}$
are already $m$-guillotine, since the recursive cutting that exhibits
that $E$ is $m$-guillotine as a network also serves to establish that
the set ${\cal D}$ of disks is $m$-guillotine as a set of disks.
Specifically, if the $m$-span along a cut $\ell$ of a rectangle $\rho$
is contained in the edge set $E$, then there are at most $O(m)$ points
within $\rho$ where $\ell$ crosses edges of $E$, and, in particular,
at most $O(m)$ points where $\ell$ properly crosses diameter segments
of disks.  This implies that the segment $\ell\cap \rho$ fully crosses
at most $O(m)$ disks, implying that $\ell\cap \rho$ intersects $O(m)$
disks (Lemma~\ref{lem:6} implies that each endpoint of $\ell\cap \rho$
intersects at most 6 disks), making $\ell$ an $m$-perfect cut (for
appropriate choices of constants $c_1$ and $c_2$).

If $E'\neq E$, then consider an $m$-span segment, $ab$, that was added
to $E$ along some cut $\ell$ in the recursive partitioning of $\rho$.
Let $a'b' = \ell\cap \rho$, with $a$ the endpoint of $ab$ that is closest to $a'$ (i.e.,
the order along $\ell$ is $(a', a, b, b')$).
We then distinguish two cases:
\begin{description}
\item[(a)] If the $m$-span $ab$ does {\em not} fully cross a disk of
  ${\cal D}$, then we know that $a'b'$ intersects only $O(m)$ disks of
  ${\cal D}$, by the following reasoning.  
By the definition of $m$-span for $E$, the segments $a'a$ and $bb'$
each have at most $m$ points along them where the segment crosses
(orthogonally) an edge of $E$ corresponding to a diameter segment of a
disk; thus, each of $a'a$ and $bb'$ fully cross at most $m$ disks of
${\cal D}$.  Since we are assuming that $sb$ does not fully cross any
disk, we know that the total number of disks that are fully crossed by
$a'b'$ is $O(m)$.  Points $a'$ and $b'$ each lie within at most 6
disks, by Lemma~\ref{lem:6}.  Any disk that intersects $a'b'$ must
either contain $a'$ or $b'$ or be fully crossed by $a'b'$ (meaning
that, when going from $a'$ to $b'$ along the segment, we both enter
and leave the disk).  Thus, in total there are $O(m)$ disks of ${\cal
  D}$ that are intersected by $a'b'=\ell\cap \rho$, implying that
$\ell$ is $m$-perfect with respect to ${\cal D}$ and $\rho$.
\old{The points $a$ and $b$ each lie
within at most 6 disks of ${\cal D}$. Since $ab$ does not fully cross any disk, we
know that along the segment $ab$, when going from $a$ towards $b$, we
exit only those disks that $a$ lies within (and $b$ does not), and we
enter only those disks that $b$ lies within (and $a$ does not).  Thus,
the interior of the segment $ab$ has at most 12 crossing points (where it crosses the boundary of a disk).
In total, then, the segment $a'b'$ has at most $2m+12$ crossing points along it, and the segment $a'b'$ intersects at most
$m+12$ disks.
%
%%% xxx Issue:  The segment $a'b'$ can fully cross disks whose (large) bounding boxes are not crossed at all by $a'b'$.
% Change to m-spans of the circles, instead of the squares....
%%  Another option:  Put a ``+'' through the middle of the disks/squares.
%%  Then, a vertical segment fully crosses the disk iff it intersects the plus.  This works too, if we want discrete coordinates.
%
Thus, the cut $\ell$ is
$m$-perfect with respect to the disks ${\cal D}$ (for constants $c_1\geq 1$, $c_2\geq 12$ in the definition of $m$-perfect).
% provided that $m+12\leq 2m$, i.e., $m\geq 12$.  (Alternatively, the definition of
%   $m$-perfect could have used a bound of $m+12$, or indeed any $c_1
%   m+c_2$ for constants $c_1,c_2$.)
% xxx rephrase??
}

\item[(b)] If the $m$-span $ab$ {\em does} fully cross at least one
  disk, then such a disk has the property that the projection of its
  center point onto $\ell$ lies on the segment $ab$.  Assume, without
  loss of generality, that $\ell$ (and thus $ab$) is vertical.  Then,
  among disks that $ab$ fully crosses, if there are any with center
  point to the left of $\ell$, let $D_L$ (centered at $p_L$, with
  radius $r_L$) be a disk with leftmost center point $p_L$.
  Similarly, if there are fully crossed disks with center point to the
  right of $\ell$, let $D_R$ (centered at $p_R$, with radius $r_R$) be
  a disk with rightmost center point $p_R$.

\begin{claim}
\label{claim:1}
If there are any disks of ${\cal D}$ that are fully crossed by $ab$
and have center point to the left of $ab$, then the disk of radius
$r_L+2|ab|$ centered at $p_L$ covers any disk $D\in {\cal D}$, with
center $p$ and radius $r$, that is fully crossed by $ab$ and has $p$
to the left of $ab$.  Similarly, if there are any disks of ${\cal D}$
that are fully crossed by $ab$ and have center point to the right of
$ab$, then the disk of radius $r_R+2|ab|$ centered at $p_R$ covers any
disk of ${\cal D}$ that is fully crossed by $ab$ and has its center to
the right of $ab$.
\end{claim}

\begin{proof}
Assume that there is at least one disk of ${\cal D}$ that is fully
crossed by $ab$ and has its center to the left of $ab$.  Then, we let
disk $D_L$, with center $p_L$ and radius $r_L$, be one such disk whose
center point is leftmost.  We know that the center point $p_L$
projects onto $ab$; let $h_L$ be the distance from $p_L$ to $ab$
(i.e., the (horizontal) distance from $p_L$ to its projection on the
line $\ell$).  Let $D\in {\cal D}$ be a disk, with center $p$ and
radius $r$, that is fully crossed by $ab$, with center point $p$ to
the left of $ab$, and let $h$ be the (horizontal) distance from $p$ to
$ab$.  By our choice of $D_L$, we know that $h_L\geq h$.  Let $q\in D$
be any point in disk $D$.  Refer to Fig.~\ref{fig:claim}. Then, by
the triangle inequality,
$$|p_Lq| \leq |p_L p| + |pq|.$$
Also by the triangle inequality, since the points $p_L$ and $p$
project onto $ab$ at two points whose $y$-coordinates differ by at
most $|ab|$, we get
$$|p_L p| \leq |x(p_L)-x(p)| + |y(p_L) - y(p)| \leq (h_L-h) + |ab|,$$
where $x(\cdot)$ and $y(\cdot)$ denote $x$- and $y$-coordinates.
Thus, using the fact that $|pq|\leq r$ (since $q\in D$), we get that 
$$|p_Lq| \leq |p_L p| + |pq| \leq (h_L-h) + |ab| + r.$$
Again using the triangle inequality, we know that $|pa| \leq h +
|ab|$; thus, $h\geq |pa|-|ab|$.  Since disk $D$ is fully crossed by
$ab$, so that $a$ is not contained in $D$, we know that $|pa|\geq r$,
implying that
$$r-h \leq r - (|pa|-|ab|) \leq |ab|.$$
Putting this in the above inequality, we get
$$|p_Lq| \leq (h_L-h) + |ab| + r \leq h_L +|ab| +|ab| \leq r_L + 2|ab|,$$
implying that any $q\in D$ is at distance at most $r_L+2|ab|$ from the
center point $p_L$, as we claimed.

The proof for disks with center points to the right of $ab$ is symmetric.
\end{proof}

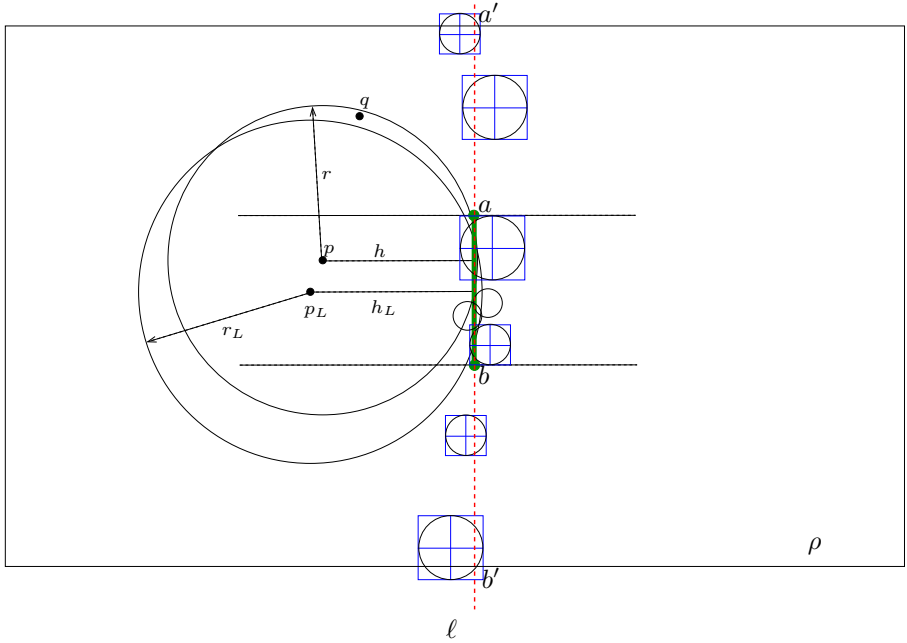
\begin {figure}[htbp!]%[tb]
\centering
\begin{picture}(0,0)%
\includegraphics{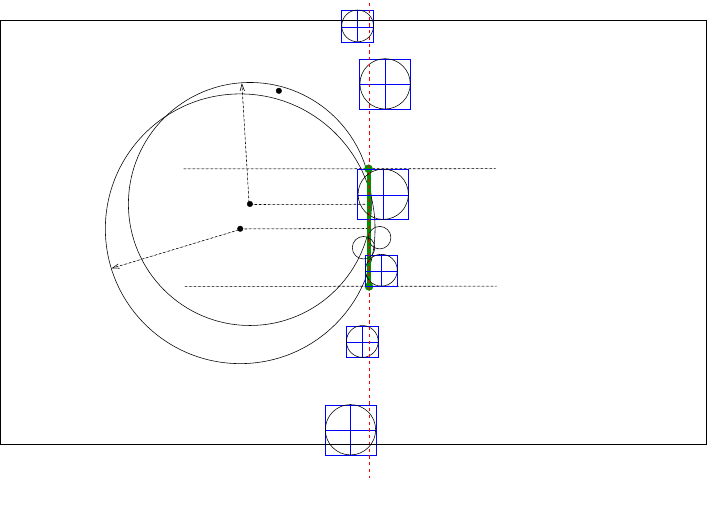}%
\end{picture}%
\setlength{\unitlength}{1184sp}%
\begingroup\makeatletter\ifx\SetFigFont\undefined%
\gdef\SetFigFont#1#2#3#4#5{%
  \reset@font\fontsize{#1}{#2pt}%
  \fontfamily{#3}\fontseries{#4}\fontshape{#5}%
  \selectfont}%
\fi\endgroup%
\begin{picture}(18869,13453)(654,-12806)
\put(9901,-12661){\makebox(0,0)[lb]{\smash{{\SetFigFont{10}{12.0}{\familydefault}{\mddefault}{\updefault}$\ell$}}}}
\put(10651,-11611){\makebox(0,0)[lb]{\smash{{\SetFigFont{10}{12.0}{\familydefault}{\mddefault}{\updefault}$b'$}}}}
\put(17476,-10786){\makebox(0,0)[lb]{\smash{{\SetFigFont{10}{12.0}{\familydefault}{\mddefault}{\updefault}$\rho$}}}}
\put(10576,239){\makebox(0,0)[lb]{\smash{{\SetFigFont{10}{12.0}{\familydefault}{\mddefault}{\updefault}$a'$}}}}
\put(10576,-3736){\makebox(0,0)[lb]{\smash{{\SetFigFont{10}{12.0}{\familydefault}{\mddefault}{\updefault}$a$}}}}
\put(10576,-7336){\makebox(0,0)[lb]{\smash{{\SetFigFont{10}{12.0}{\familydefault}{\mddefault}{\updefault}$b$}}}}
\put(8089,-1483){\makebox(0,0)[lb]{\smash{{\SetFigFont{7}{8.4}{\familydefault}{\mddefault}{\updefault}$q$}}}}
\put(8382,-4676){\makebox(0,0)[lb]{\smash{{\SetFigFont{7}{8.4}{\familydefault}{\mddefault}{\updefault}$h$}}}}
\put(8329,-5843){\makebox(0,0)[lb]{\smash{{\SetFigFont{7}{8.4}{\familydefault}{\mddefault}{\updefault}$h_L$}}}}
\put(7336,-4591){\makebox(0,0)[lb]{\smash{{\SetFigFont{7}{8.4}{\familydefault}{\mddefault}{\updefault}$p$}}}}
\put(7291,-3046){\makebox(0,0)[lb]{\smash{{\SetFigFont{7}{8.4}{\familydefault}{\mddefault}{\updefault}$r$}}}}
\put(6931,-5881){\makebox(0,0)[lb]{\smash{{\SetFigFont{7}{8.4}{\familydefault}{\mddefault}{\updefault}$p_L$}}}}
\put(5206,-6316){\makebox(0,0)[lb]{\smash{{\SetFigFont{7}{8.4}{\familydefault}{\mddefault}{\updefault}$r_L$}}}}
\end{picture}%

%%%
\caption{ Illustration of the proof of Claim~\ref{claim:1}.  The
  vertical cut $\ell$ intersects the rectangle $\rho$ in the segment
  $a'b'$.  The $m$-span ($m=4$), $ab$, is shown in thick green.  The
  disk $D_L$, centered at $p_L$, of radius $r_L$, is shown, as is
  another disk $D$, centered at $p$, with radius $r$.  The claim is
  that any point $q\in D$ must be covered by an enlarged disk, of
  radius $r_L+2|ab|$, centered on $p_L$.
\label{fig:claim}}
\end{figure}

We now modify the set ${\cal D}$ of disks as follows.  If there are
disks fully crossed by $ab$ with center to the left of $ab$, we select
one, $D_L$, having leftmost center point, and increase its radius to
$r_L+2|ab|$.  By Claim~\ref{claim:1}, we know that this new enlarged
disk covers all disks $D_i \in {\cal D}$ that were fully crossed by
$ab$, with center points $p_i$ to the left of $ab$.  This implies that
we can remove all such disks $D_i$ (i.e., shrink $r_i$ to zero), while
maintaining connectivity of the set of disks.  If there are disks
fully crossed by $ab$ whose center points are to the right of $ab$, we
similarly enlarge one of them, $D_R$, whose center point $p_R$ is
rightmost, while removing the others (shrinking their radii to zero).
The net change to the sum of the radii of all disks is that it goes up
by at most~$O(|ab|)$.

As in case (a) above, we know that the segments $a'a$ and $bb'$ each
intersect $O(m)$ disks.  After the modification above, we know that
$ab$ intersects at most $O(1)$ disks (at most 2 (enlarged) disks of
radii $r_L+2|ab|$ and $r_R+2|ab|$, and at most 6 disks containing $a$
and 6 disks containing $b$).  Thus, in total $a'b'=\ell\cap \rho$
intersects at most $O(m)$ disks of ${\cal D}$, so we know that $\ell$
is $m$-perfect with respect to ${\cal D}$ and $\rho$.
%%%
\old{
have $m$ points of intersection with disk boundaries (by the
definition of $m$-span for $E$), and that the points $a$ and $b$ each
lie within at most 6 disks. By the modification process above, all of
the disks of the set ${\cal D}$ that are fully crossed by $ab$ are
replaced by at most two disks (the enlarged disks, of radii
$r_L+2|ab|$ and $r_R+2|ab|$).  Thus, after the modification, the
number of disks intersected by $a'b'$ is at most $m+6+2+6=m+14$.
Thus, the cut $\ell$ is $m$-perfect with respect to the disks ${\cal
  D}$ (for constants $c_1\geq 1$, $c_2\geq 14$ in the definition of $m$-perfect).
% provided that $m+14\leq 2m$, i.e., $m\geq 14$.     %% if the definition uses 2m for m-perfect
}
\end{description}

\old{
\begin {figure}[htbp!]%[tb]
\centering
\begin{picture}(0,0)%
\includegraphics{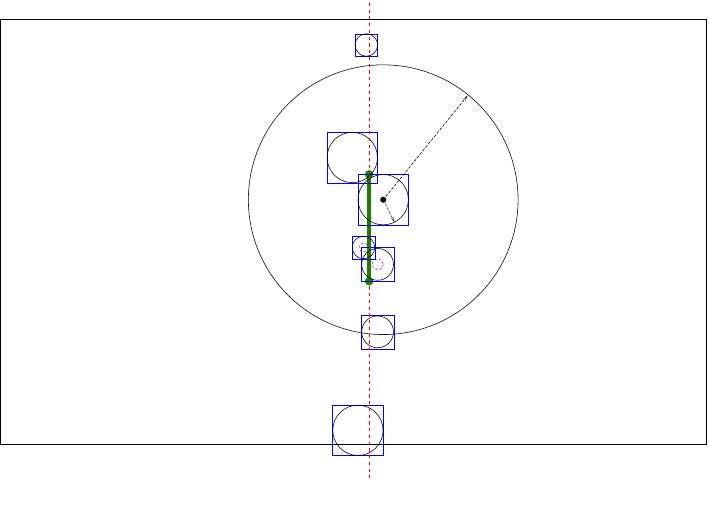}%
\end{picture}%
\setlength{\unitlength}{1184sp}%
\begingroup\makeatletter\ifx\SetFigFont\undefined%
\gdef\SetFigFont#1#2#3#4#5{%
  \reset@font\fontsize{#1}{#2pt}%
  \fontfamily{#3}\fontseries{#4}\fontshape{#5}%
  \selectfont}%
\fi\endgroup%
\begin{picture}(18869,13458)(654,-12811)
\put(10576,239){\makebox(0,0)[lb]{\smash{{\SetFigFont{10}{12.0}{\familydefault}{\mddefault}{\updefault}$a'$}}}}
\put(9901,-12661){\makebox(0,0)[lb]{\smash{{\SetFigFont{10}{12.0}{\familydefault}{\mddefault}{\updefault}$\ell$}}}}
\put(10651,-11611){\makebox(0,0)[lb]{\smash{{\SetFigFont{10}{12.0}{\familydefault}{\mddefault}{\updefault}$b'$}}}}
\put(17476,-10786){\makebox(0,0)[lb]{\smash{{\SetFigFont{10}{12.0}{\familydefault}{\mddefault}{\updefault}$\rho$}}}}
\put(12151,-3511){\makebox(0,0)[lb]{\smash{{\SetFigFont{7}{8.4}{\familydefault}{\mddefault}{\updefault}$r_i+|ab|$}}}}
\put(9826,-3886){\makebox(0,0)[lb]{\smash{{\SetFigFont{10}{12.0}{\familydefault}{\mddefault}{\updefault}$a$}}}}
\put(9751,-7186){\makebox(0,0)[lb]{\smash{{\SetFigFont{10}{12.0}{\familydefault}{\mddefault}{\updefault}$b$}}}}
\put(11026,-4861){\makebox(0,0)[lb]{\smash{{\SetFigFont{7}{8.4}{\familydefault}{\mddefault}{\updefault}$r_i$}}}}
\put(10651,-4411){\makebox(0,0)[lb]{\smash{{\SetFigFont{7}{8.4}{\familydefault}{\mddefault}{\updefault}$p_i$}}}}
\end{picture}%

%%%
\caption{
    An example of a $4$-span $ab$, shown in thick green, along vertical cut $\ell$. The $4$-span fully crosses the bounding square of the disk of radius $r_1$ centered on $p_1$.  Thus, we expand its radius to $r_1+|ab|$, and shrink the radii of the other two disks that intersect $ab$, to the (dashed) magenta circles.  In this way, the set of disks that are intersected by $ab$ are replaced with one larger disk (centered on $p_i$, of radius $r_i+|ab|$), and a set of smaller disks that do not cross the cut. 
\label{fig:m-span}}
\end{figure}
}

Applying the above modification of the disk radii to all cuts in the
recursive set of cuts associated with the $m$-guillotine edge set
$E'$, we obtain a set of disks ${\cal D}_m$.  Since the sum of all
$m$-span lengths $|ab|$ is charged off to the length $\lambda(E)=O(\sum_i r_i)$ of the
edge set $E$, totalling at most $O(1/m)$ times the sum $\sum_i r_i$,
we have shown that the sum of the
radii of disks of ${\cal D}_m$ is at most $(1+(C/m))\sum_i r_i$.

Finally, by Lemma~\ref{lem:round}, we can round up the radii of the
resulting disks ${\cal D}_m$ to make them ${\cal R}$-disks, while
increasing the sum of radii by at most a factor $1+O(1/m)$.
\end{proof}

% 9/6/14: added this observation
% wait, this holds only for optimal, or nearly optimal, sets of disks:
\old{
Next, we observe that we can assume that no point of the plane lies
within more than $O(1)$ disks of ${\cal D}$.  (If some point $p$
lies within more than 6 disks of ${\cal D}$, then we know that two
of them have centers within a 60-degree cone with apex $p$; thus, the
distance between the two centers is less than the larger of the two
radii, implying that we can shrink the radius of the larger disk, so
that $p$ no longer lies within it, while keeping the set of disks 
connected.)
}

\old{
Now, let $\rho$ denote the axis-aligned bounding box of the resulting
disks.  If there are no disks strictly interior to $\rho$ (i.e.,
if all disks intersecting $\rho$ intersect (touch) the boundary,
$\partial\rho$, of $\rho$), then we are done: the set of disks is,
by definition, $m$-guillotine already.  Thus, we assume that there is
at least one disk of ${\cal D}$ strictly interior to $\rho$.

{\em Case 1:} If there exists an $m$-perfect cut of $\rho$, by a horizontal or
vertical line intersecting at most $2m$ disks interior to $\rho$, then we
partition $\rho$ with it and recurse on the two boxes on each side of
the cut.  We observe that if such an $m$-perfect cut exists, then one
exists with the property that the defining coordinate of the cut is
from the set ${\cal I}_x$ (for vertical cuts) or ${\cal I}_y$ (for
horizontal cuts), since we can translate the cut between two
consecutive such coordinates without changing the set of disks it
intersects.  
%
% 9/6/14: removed: I don't think this is necessary
% (More precisely, we can translate a cut to be
% infinitesimally close to one of the discrete coordinates without
% changing its combinatorial type.  To address this technicality, we can
% either augment the sets ${\cal I}_x$ and ${\cal I}_y$ with the
% midpoints of the intervals between consecutive coordinates, or we can
% replace each coordinate $z$ with three coordinates -- $z$ and $z^+$
% (infinitesimally greater than $z$) and $z^-$ (infinitesimally smaller
% than $z$).)

{\em Case 2:} Thus, we now assume that no $m$-perfect cut exists for partitioning
$\rho$, and that $\rho$ contains in its interior at least one disk.
Consider the set of bounding boxes (squares) of the disks that are interior to $\rho$;
let $E$ be the set of boundary edges of these squares.  The edges $E$ form a network of
horizontal/vertical segments interior to $\rho$, of total length $\lambda(E)$.  
By the usual $m$-guillotine argument~\cite{m-gsaps-99}, we know that there exists a vertical
or horizontal cut, $\ell$, through $\rho$ such that 
}

\old{
In this case, we would like to show that a {\em favorable} cut exists.
In order to define such cuts, we need to introduce two important
notions: the {\em cost} of a cut and the {\em chargeable length} 
of a cut.

For a vertical line, $\ell_x$, through coordinate $x$, let $f(x)$
denote the length of the {\em $m$-span} of $\ell_x$ with respect to
${\cal D}$ and $\rho$: $f(x)=0$ if, within $\rho$, $\ell_x$ intersects
at most $2m$ disks of ${\cal D}$; otherwise, if $\ell_x$ intersects
$K>2m$ disks of ${\cal D}$ within $\rho$, then $f(x)$ is the distance
(along $\ell_x$) from the first point, $a_m$, where $\ell_x$ enters
into the (axis-aligned) bounding box of the $m$th disk, going from the top boundary of $\rho$ downwards
along $\ell_x$, to the first point, $b_m$, where $\ell_x$ enters into
the bounding box of the $m$th disk, going from the bottom boundary of $\rho$ upwards along
$\ell_x$.  Because $K>2m$, we know that $a_m$ is above $b_m$ and that
there must be at least one disk of ${\cal D}$ whose bounding box
intersects $\ell_x$ is a subset of the {\em bridge segment} $a_mb_m$.
We similarly define $g(y)$ to be the length of the (horizontal)
bridge segment along a horizontal cut $\ell_y$
through coordinate $y$.  

We think of $f(x)$ and $g(y)$ as the cost of augmentation for the
network ${\cal N}_\rho$ consisting of the union of disks, truncated
within $\rho$, that are the boundaries of the disks ${\cal D}$; by
adding segments (bridges) of length $f(x)$ (resp., $g(y)$), a vertical
cut $\ell_x$ (resp., horizontal cut $\ell_y$) can be made $m$-perfect.

We claim that we can charge off the lengths of the bridges that would
suffice to augment the network ${\cal N}_\rho$ to make it
$m$-guillotine, in the usual sense of an $m$-guillotine network
(subdivision), as in \cite{m-gsaps-99}.  Specifically, we argue
that we can select cuts for which the bridge lengths ($m$-spans) can
be charged off to the total length of the network (sum of the disk 
circumferences, which is $O(\sum_i r_i)$), showing that the sum of the
bridge lengths is at most $(C/m)\sum_i r_i$.

We partition each disk (bounding the disks ${\cal D}$) into four
90-degree arcs: two ``vertical arcs'' (with angular ranges (-45,45)
and (135,225)) and two ``horizontal arcs'' (with angular ranges
(45,135) and (225,315)).  We define the ``chargeable length'' of a
vertical cut $\ell_x$ to be the ``$m$-dark'' length of $\ell_x\cap
\rho$.  Specifically, a subsegment $ab$ of $\ell_x\cap \rho$ is said
to be $m$-dark with respect to ${\cal N}_\rho$ if for any $p\in ab$, the
rightwards and leftwards rays from $p$ each cross at least $m$
vertical arcs of ${\cal N}_\rho$ before exiting $\rho$.  If we cut
$\rho$ along $\ell_x$, then the $m$-dark portion of the cut can be
charged off to the left/right sides of the vertical arcs of ${\cal
N}_\rho$ lying to the right/left of $\ell_x$, distributing the charge
to be ($1/m$)th to each of the $m$ arcs first hit. 

A vertical cut $\ell_x$ or horizontal cut $\ell_y$ for $\rho$ is favorable if  the
chargeable length of the cut is at least as long as the cost ($f(x)$ or $g(y)$) of the cut.

\begin{lemma}
For any network ${\cal N}_\rho$, a favorable cut exists.
\end{lemma}
\begin{proof}
Our charging scheme is based on the observation, following the method
of \cite{m-gsaps-99}, that there must exist a favorable vertical
cut $\ell_x$ or horizontal cut $\ell_y$ for $\rho$ such that the
chargeable length of the cut is at least as long as the cost ($f(x)$
or $g(y)$) of the cut.  The existence of a favorable cut follows from
the observation that $\int_{x\in \rho} f(x) dx = \int_{y\in \rho} h(y)
dy$, where $h(y)$ is the chargeable length associated with the
horizontal cut $\ell_y$; thus, assuming, without loss of generality,
that $\int_{x\in \rho} f(x) dx \geq \int_{y\in
\rho} g(y) dy$, we see that there must exist a value $y^*$ where
$g(y^*)\leq h(y^*)$, which defines a favorable horizontal cut
$\ell_{y^*}$ for which the chargeable length exceeds the length of the
$m$-span.  (In case $\int_{x\in \rho} f(x) dx <
\int_{y\in \rho} g(y) dy$, there exists a favorable vertical cut.)
Further, we claim that there must exist a favorable cut corresponding
to the coordinate sets ${\cal I}_x$, ${\cal I}_y$: the chargeable
length of cut $\ell_{y^*}$ does not change as we perturb $y^*$ to the
nearest coordinate in the set ${\cal I}_y$, while, by convexity of the
circular arcs, the length of the $m$-span associated with a cut is
locally minimized at endpoints of the intervals defined by the points
of ${\cal I}_y$.
\end{proof}

Once a favorable cut is found with respect to rectangle $\rho$, the
cut partitions the problem into two subrectangles, and the argument is
recursively applied to each.  Since each circular arc of length
$\lambda$ is charged for length at most $\lambda/2m$ on each of its
two sides, we get that the overall length of all $m$-spans that are
associated with favorable cuts constructed recursively in converting
the network ${\cal N}_{BB({\cal D})}$ to an $m$-guillotine network is
at most $(C/m)\sum_i r_i$, for a constant $C$.

Finally, we claim that the disks ${\cal D}$ can be converted to an
$m$-guillotine set ${\cal D}_m$, having sum of radii at most
$(C/m)\sum_i r_i$: Associated with each $m$-span $a_mb_m$ that is
added to the network of circular arcs in order to make the network
$m$-guillotine, we enlarge the radius of the disk defining one of the
$m$-span endpoints (say, $a_m$), by at most $|a_mb_m|$, so that this
disk now covers the entire $m$-span segment.  Since we have enlarged
one disk in a set of disks whose union is connected, the union remains
connected.  Thus, with a total increase in radii of at most
$(C/m)\sum_i r_i$, we end up with an $m$-guillotine set ${\cal D}_m$
of disks, proving our structure theorem. 
\end{proof}
}

\subsection{Dynamic Programming}
We now give an algorithm to compute a minimum-cost (sum of radii)
$m$-guillotine set of ${\cal R}$-disks whose union is connected.  The
algorithm is based on dynamic programming.  A subproblem is specified
by a rectangle, $\rho$, with $x$- and $y$-coordinates among the sets
${\cal I}_x$ and ${\cal I}_y$, respectively, of discrete coordinates.
The subproblem includes specification of {\em boundary information},
for each of the four sides of $\rho$.  Specifically, the boundary
information includes:
\begin{description}
\item[(i)] $O(m)$ ``portal disks'', which are ${\cal
R}$-disks intersecting the boundary, $\partial \rho$, of $\rho$, with
at most $O(m)$ disks specified per side of $\rho$; and, 
\item[(ii)] a connection pattern, specifying which subsets of the portal disks are
required to be connected within $\rho$.  (Thus, the connection pattern is specified by
giving a partition of the set of portal disks; while the number of partitions is exponential in $m$,
it is constant for fixed $\epsilon=\lceil 2/m\rceil$.)
\end{description}
There are a polynomial number of subproblems (specifically,
$n^{O(m)}$), for any fixed $m$.  For a given subproblem, the dynamic
program optimizes over all (polynomial number of) possible cuts $\ell$
(horizontal at ${\cal I}_y$-coordinates or vertical at ${\cal
  I}_x$-coordinates), and choices of up to $O(m)$ ${\cal R}$-disks
intersecting the cut segment $\ell\cap \rho$, along with all possible
compatible connection patterns for each side of the cut.  (The optimal
value of the objective function for each of the two corresponding
subproblems on each side of each candidate cut $\ell$ have been
precomputed and tabulated, since we fill in the table of data in order
of increasing (combinatorial) size of the rectangles $\rho$.)  The
result is an optimal $m$-guillotine set of ${\cal R}$-disks such that
their union is connected and the sum of the radii is minimum possible
for $m$-guillotine sets of ${\cal R}$-disks.  Since we know, from the
structure theorem, that an optimal set of disks centered at points $P$
can be converted into an $m$-guillotine set of ${\cal R}$-disks
centered at points of $P$, whose union is connected, and we have
computed an optimal such structure, we know that the disks obtained by
our dynamic programming algorithm yield an approximation to an optimal
set of disks.  In summary, we have shown the following result:

\begin{theorem}\label{thm:PTAS}
For any $\epsilon > 0$, there is an approximation algorithm for CRA
running in time $n^{O(1/\epsilon)}$ that produces an approximate
solution within factor $(1+\epsilon)$ of the optimal.
\end{theorem}

%%%% Joe removed this subsection for the revised CGTA submission, Aug 2, 2016:  
\old{
Joe wrote to Sandor:
On a related note: what do you think about the ``omission'' I did?  (I removed the short subsection about the bounded radius PTAS, for the
special case (assumed that there was a connected path of disks between any two points, with sum of radii a constant factor greater than the distance between the two points, for any pair of points))
The issue is this:  Adding the details to the level we have now for the other PTAS would be very time consuming, with marginal payoff (since it still does not solve the ``real'' question, without the extra assumption).
However, omitting it means that what we currently have is a PTAS for the unbounded radius case, for which we do not even have a hardness proof.....
(this raises the question of whether the hardness proof for geometric bounded radius case satisfies the ``assumption'' that was needed for the PTAS: ie, 
does that PTAS solve a problem we know to be hard?)

I fully believe the PTAS, but to make the details complete and convincing would take
too much time/effort, and it only solves a special case.
Instead, I will try later to solve the general case, without the special assumption,
and that result may be able to stand on its own.}
\old{
\subsection{A PTAS for CRA with Bounded Radii}

We now address the case of bounded radii, in which disk $i$ has a
maximum allowable radius, $\bar r_i<\infty$.  The PTAS given above
relied on disk radii being arbitrarily large, so that we could
increase the radius of a single disk to cover the entire $m$-span
segment $ab$.  A slightly different argument is needed for the case of bounded
radii.

% 9/21/14:  attempts to solve the more general problem, without special assumptions:
\old{
Look at a cut, and consider the m-span.
Look at what COULD be covered if every disk that currently intersects the m-span were increased
to its maximum possible size.  This may give several components (only one in the unbounded radius case).
We will consider the m-span to be ``maximally covered'' by disks if they cover all of this portion of the span.
Claim: we can augment radii, totalling an increase of at most O(|m-span|), so that the m-span is maximally covered.

In the DP, we just specify that the m-span is maximally covered.
Wait!!  This does not yet work:  issue is that the maximal coverage could include a very large radius
disk that barely touches the m-span, and this is not even touched/covered by any of the disks of the input
set.  So, to cover as much as maximal does, means that we may have to grow a HUGE disk??

9/24/14: idea: Claim: we can make maximal with respect to all disks of radii O(|ab|)
that potentially meet the m-span ab.  There could still be many other disks of larger radius
that meet ab at various nearly-isolated points, but we do not need to know about all of them
for purposes of solving the subproblem.
Make this more precise!

Look at the set U that is the union of all largest disks.
(issue: there could be one point far away with an infinite radius.... that make U trivially cover everything!)
Rather, look at the graph G that connects 2 points with an edge iff their maximum-radius disks
intersect.  
Consider G with Euclidean lengths on edges.
}

We obtain a PTAS for the bounded radius case, if we make an additional
assumption: we assume that for any segment $pq$, with $p$ and $q$ each within
disks of allowable radii, there exists a connected set of disks,
centered at points of $p_i\in P$ and having allowable radii $r_i\leq \bar r_i$,
such that $p$ and $q$ each lie within the union of the disks and the
sum of the radii of the disks is $O(|pq|)$.

We highlight only the differences
from the unbounded radius case. The PTAS method proceeds as above in
the unbounded radius case, except that we now modify the proof of the
structure theorem by replacing each $m$-span bridge $a_mb_m$ by a
shortest connected path of allowable ${\cal R}$-disks.  
We call such a set of disks a {\em bridging disk path}. 
%% xxx add more, if we put it back
We know, from our
additional assumption, that the sum of the radii along such a shortest
path is $O(|a_mb_m|)$, allowing the charging scheme to proceed as
before.  The dynamic programming algorithm changes some as well, since
now the subproblem specification must include the bridging
disk-path, which is specified only by its first and last disk (those
associated with the bridge endpoints $a_m$ and $b_m$); the path
itself, which may have complexity $\Omega(n)$, is implicitly
specified, since it is a shortest path (which we can assume to be
unique, since we can specify a lexicographic rule to break ties).
%
%% among all paths, pick one to min r_1, then r_2, etc.  This makes a ``canonical path''. Describe a subproblem in full detail?  xxx

In summary, we have

\begin{theorem}\label{thm:PTAS-bounded}
There is a PTAS for CRA with
bounded radii, assuming that for any segment $pq$, with $p$ and
$q$ within feasible disks, there exists a (connected) path of feasible
disks whose sum of radii is $O(|pq|)$.
\end{theorem}

}
%%% end of removed subsection.

\section{Experimental Results}\label{sec:exp}

\begin {figure}[h!]%[tb]
  \centering
  \includegraphics[width=0.6\columnwidth]{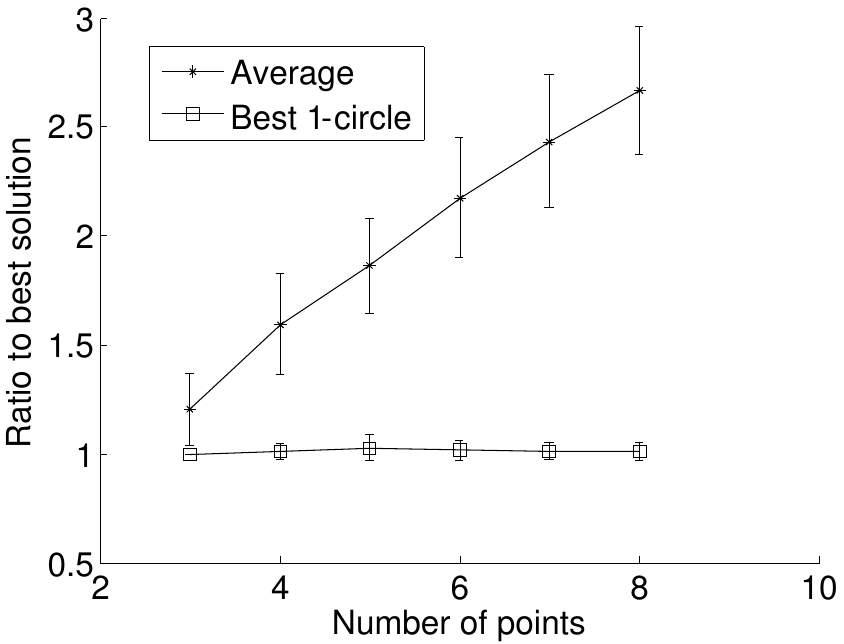}
  \caption{
    Ratios of the average 
    over all enumerated trees and of the best 1-disk tree
    to the optimal $\sum r_i$.  Results were 
    averaged over 100 trials for each number. 
    \label{exp_ratios}
  }
\end{figure}

It is interesting that even in the worst case, a one-disk
solution is close to being optimal. This is supported
by experimental evidence. 
Using the MATLAB programming language, we analyzed the difficulty 
associated with finding the optimal assignment of radii in randomized trials.  
In order to generate
random problem instances, we considered different
numbers of points, uniformly distributed in a two-dimensional circular region. 
For each trial considering a single distribution of points, we enumerated all possible
spanning trees using the method described in~\cite{Avis96_reverse}
and recorded the optimal value with the algorithm described in Section 2. 
Additionally, for each trial, we noted the sum of $R$ values obtained from the best one-disk solution.  
Fig.~\ref{exp_ratios} shows that the ratio of the best one-disk solution to the best solution remains relatively flat as
the number of points increases, as opposed to the average solution, suggesting that the one-disk solution is an excellent
heuristic choice. These results were obtained in less than 6 hours using an i7 PC.

\section{Conclusion}\label{sec:conc}

A number of open problems remain. One of the most puzzling is the issue of
complexity in the absence of upper bounds on the radii.
The strong performance of the one-disk solution (and even better
of solutions with a higher, but limited, number of disks), and the difficulty
of constructing solutions for which the one-disk solution is not optimal
strongly hint at the possibility of the problem being polynomially solvable.
Another indication is that our positive results for one or two disks 
only needed triangle inequality, i.e., they did not explicitly make
use of geometry. 

One possible approach to this problem may be to use methods from linear programming. Modeling
the objective function and the variables is straightforward;
describing the connectivity of a spanning tree by linear cut constraints
is also well known. However, even though separating over 
the exponentially many cut constraints is polynomially solvable (and hence optimizing
over the resulting polytope), the overall polytope is not necessarily integral.
On the other hand, we have been unable to prove NP-hardness without upper bounds
on the radii, even in the more controlled context of graph-induced distances.
Note that some results were obtained by means of linear programming:
the tight lower bound for 2-disk solutions (shown in Fig.~\ref{2_disk_strict}) % Fig.~7 
was found by solving appropriate LPs.
% yyy please confirm I made correct change:  It previouosly has ``Fig. 7'' rather than a ref to the figure label. The current Fig. 7 seems to give a lower bound on the 3/2 approx factor for the one-disk solution??  I suspect we mean Figure 11 ( 2_disk_strict  ) ??
% Figure 7 = 		\caption{A lower bound of $\frac{3}{2}$ for 1-disk solutions.\label{1_disk_strict}}
% Figure 11 = 	\caption{A lower bound of $\frac{5}{4}$ for 2-disk solutions.\label{2_disk_strict}}

Other open problems are concerned with the worst-case performance of heuristics
using a bounded number of disks. We showed that two disks suffice for
a $\frac{4}{3}$-approximation in general, and a $\frac{5}{4}$-approximation
on a line; we conjecture that the general performance guarantee can be
improved to $\frac{5}{4}$, matching the existing lower bound. Obviously,
the same can be studied for $k$ disks, for any fixed $k$; at this point,
the best lower bounds we have are $\frac{7}{6}$ for $k=3$ and 
$1+\frac{1}{2^{k+1}}$ for general $k$.
We also conjecture that the worst-case ratio $f(k)$ of a best $k$-disk solution
approximates the optimal value arbitrarily well for large $k$, i.e.,
 $\lim_{k\to\infty}f(k) = 1$.

We have given a PTAS for CRA in two dimensions with unbounded radii, a
problem with unknown complexity -- there might be a polynomial-time
exact solution.  Does there exist a PTAS for the general case of CRA
with bounded radii, a problem we have shown to be NP-hard?  In the
conference version of the paper~\cite{cfh-cscmsr-11}, we sketched an
approach for how to modify our PTAS for unbounded radii to address a
special case of the CRA problem with bounded radii, if an additional
assumption is made, that for any segment $pq$, with $p$ and $q$ within
feasible disks, there exists a (connected) path of feasible disks
whose sum of radii is $O(|pq|)$.  It would be interesting to obtain a
PTAS for the CRA for the general case of bounded radii.

%\subsubsection*{Acknowledgments.} The heading should be treated as a
%subsubsection heading and should not be assigned a number.
\section*{Acknowledgments}

We thank Ferran for being a great inspiration to all of us.
A preliminary version of this work appears in the Algorithms and Data Structures Symposium (WADS), 2011~\cite{cfh-cscmsr-11}.
This work was started during the 2009 McGill/INRIA/University of Victoria Bellairs Workshop on Computational Geometry.
We thank all other participants for contributing to the great atmosphere.
This work has been partially supported by the National Science Foundation (grants CCF-1054779 and IIS-1319573, Erin Chambers; grants CCF-1018388 and CCF-1526406, Joseph Mitchell), by the Binational Science Foundation (BSF 2010074, Joseph Mitchell), Sandia National Labs (Joseph Mitchell), and three individual NSERC Discovery grants (one each
for Venkatesh Srinivasan, Ulrike Stege, and Sue Whitesides).

%\small
\section*{References}
\bibliographystyle{abbrv}
\bibliography{lit,refs}

\end{document}